\documentclass[11pt]{article}
\usepackage{times,amsfonts,amssymb,amsmath,amsthm,epsfig,graphicx}
\usepackage{times} \usepackage{multirow,xspace}
\usepackage{subfigure} 
\usepackage{setspace} 

\usepackage[marginpar]{todo} \usepackage{hyperref,url} \usepackage[margin=1in]{geometry}

\usepackage{graphicx}

\newtheorem{theorem}{Theorem}[section]

\newtheorem{claim}{Claim}[section]

\newtheorem{corollary}[theorem]{Corollary}

\newtheorem{lemma}[theorem]{Lemma}

\newtheorem{observation}[theorem]{Observation}

  \newcommand{\abs}[1]{\left| #1 \right|}

\newcommand{\packing}[2]{\mathcal{#1}_{#2}}
\newcommand{\ratio}[1]{\rho_{#1}}
\newcommand{\mar}[1]{\delta_{#1}}

\date{}

\renewcommand{\labelitemi}{}

\begin{document}


\title{Parametric Packing of Selfish Items and the Subset Sum Algorithm}

\author{
Leah Epstein\thanks{ Department of Mathematics, University of
Haifa, 31905 Haifa, Israel. {\tt lea@math.haifa.ac.il}. }\and
Elena Kleiman \thanks{ Department of Mathematics, University of
Haifa, 31905 Haifa, Israel. {\tt elena.kleiman@gmail.com}. }
\and
Juli\'{a}n Mestre \thanks{ Max-Planck-Institut f\"ur Informatik, 66123
  Saarbr\"ucken, Germany.
{\tt jmestre@mpi-inf.mpg.de.} Research supported by an Alexander von Humboldt Fellowship.}
}\maketitle

\vspace{-0.7cm}

{\tiny
\begin{abstract}
  The subset sum algorithm is a natural heuristic for the classical Bin
  Packing problem: In each iteration, the algorithm finds among the unpacked
  items, a maximum size set of items that fits into a new bin. More than 35
  years after its first mention in the literature, establishing the worst-case
  performance of this heuristic remains, surprisingly, an open problem.

  Due to their simplicity and intuitive appeal, greedy algorithms are the
  heuristics of choice of many practitioners. Therefore, better understanding
  simple greedy heuristics is, in general, an interesting topic in its own
  right. Very recently, Epstein and Kleiman \emph{(Proc. ESA 2008, pages
    368-380)} provided another incentive to study the subset sum algorithm by
  showing that the Strong Price of Anarchy of the game theoretic version of
  the bin-packing problem is \emph{precisely} the approximation ratio of this
  heuristic.

  In this paper we establish the exact approximation ratio of the subset sum
  algorithm, thus settling a long standing open problem. We generalize this
  result to the parametric variant of the bin packing problem where item sizes
  lie on the interval $(0, \alpha]$ for some $\alpha \leq 1$, yielding
  tight bounds for the Strong Price of Anarchy for all $\alpha \leq
  1$. Finally, we study the pure Price of Anarchy of the parametric Bin
  Packing game for which we show nearly tight upper and lower bounds for all
  $\alpha \leq 1$.

\end{abstract}}

\section{Introduction}


\noindent{\bf Motivation and framework.} The emergence of the Internet and its
rapidly gained status as the predominant communication platform has brought up
to the surface new algorithmic challenges that arise from the interaction of
the multiple self-interested entities that manage and use the network.  Due to
the nature of the Internet, these interactions are characterized by the
(sometimes complete) lack of coordination between those entities. Algorithm
and network designers are interested in analyzing the outcomes of these
interactions. An interesting and topical question is how much performance is
lost due to the selfishness and unwillingness of network participants to
cooperate. A formal framework for studying interactions between multiple
rational participants is provided by the discipline of Game Theory.  This is
achieved by modeling the network problems as strategic games, and considering
the quality of the Nash equilibria of these games. In this paper we consider
\textit{pure} Nash equilibria and strong equilibria. These equilibria are the
result of the pure strategies of the participants of the game, where they
choose to play an action in a deterministic, non-aleatory manner.

The algorithmic problems that are usually studied from a game theoretic point
of view are abstractions of real world problems, typically dealing with basic
issues in networks. In this paper, we consider game theoretic variants of the
well-known Bin Packing problem  and its parametric version; see
\cite{CoGaJo97, G_CC_06,G_CCL_1_06} for surveys on these problems.

In the classic Bin Packing problem, we are given a set of items $I=\{ 1,2,
\ldots ,n\}$. The $i$th item in $I$ has size $s_i \in (0,1]$. The objective is
to pack the items into unit capacity bins so as to minimize the number of bins
used.  In the parametric case, the sizes of items are bounded from above by a
given value. More precisely, given a parameter $\alpha\leq 1 $ we consider
inputs in which the item sizes are taken from the interval $(0,
\alpha]$. Setting $\alpha$ to 1 gives us the standard Bin Packing problem.

As discussed in \cite{EpsteinK08}, bin packing is met in a great variety of
networking problems, such as the problem of packing a given set of packets
into a minimum number of time slots for fairness provisioning and the problem
of packing data for Internet phone calls into ATM packets, filling fixed-size
frames to maximize the amount of data that they carry. This fact motivates the
study of Bin Packing from a game theoretic perspective.  The Parametric Bin
Packing problem also models the problem of efficient routing in
networks that consist of parallel links of same bounded bandwidth between two
terminal nodes---similar to the ones considered in \cite{KP99, Bilo06,
  EpsteinK08}.  As Internet Service Providers often impose a policy which
restricts the amount of data that can be downloaded/uploaded by each user,
placing a restriction on the size of the items allowed to transfer makes the
model more realistic.


\noindent{\bf The model.} In this paper we study the Parametric Bin Packing
problem both in cooperative and non-cooperative versions. In each case the
problem is specified by a given parameter $\alpha$. The Parametric Bin Packing
game is defined by a tuple $BP(\alpha)=\langle N, (B_i)_{i\in N}, (c_i)_{i\in
  N}\rangle$. Where $N$ is the set of the items, whose size is at most
$\alpha$. Each item is associated with a selfish player---we sometimes
consider the items themselves to be the players.  The set of strategies $B_i$
for each player $i\in N$ is the set of all bins. Each item can be assigned to
one bin only. The outcome of the game is a particular assignment
$b=(b_j)_{j\in N}\in \times_{j\in N} B_j$ of items to bins. All the bins have
unit cost. The cost function $c_i$ of player $i\in N$ is defined as follows. A
player pays $\infty$ if it requests to be packed in an invalid way, that is, a
bin which is occupied by a total size of items which exceeds 1. Otherwise, the
set of players whose items are packed into a common bin share its unit cost
proportionally to their sizes.  That is, if an item $i$ of size $s_i$ is
packed into a bin which contains the set of items $B$ then $i$'s payment is
$c_i=s_i / \sum_{k \in B} s_k$. Notice that since $ \sum_{k \in B} s_k \leq 1
$ the cost $c_i$ is always greater or equal than $s_i$.  The social cost
function that we want to minimize is the number of used bins.

Clearly, a selfish item prefers to be packed into a bin which is as full as
possible. In the non-cooperative version, an item will perform an improving
step if there is a strictly more loaded bin in which it fits. At a Nash
equilibrium, no item can unilaterally reduce its cost by moving to a different
bin. We call a packing that admits the Nash conditions $NE$ packing.  We
denote the set of the Nash equilibria of an instance of the Parametric Bin
Packing game $G\in BP(\alpha)$ by $NE(G)$.

In the cooperative version of the Parametric Bin Packing game, we consider all
(non-empty) subgroups of items from $N$. The cost functions of the players are
defined the same as in the non-cooperative case. Each group of items is
interested to be packed in a way so as to minimize the costs for all group
members.  Thus, given a particular assignment, all members of a group will
perform a joint improving step (not necessarily into a same bin) if there is
an assignment in which, for each member, the new bin will admit a strictly
greater load than the bin of origin. The costs of the non-members may be
enlarged as a result of this improving step.  At a strong Nash equilibrium, no
group of items can reduce the costs of all its members by moving to different
bins. We denote the set of the strong Nash equilibria of an instance $G$ of
the Parametric Bin Packing game by $SNE(G)$. As a group can contain a single
item, $SNE(G)\subseteq NE(G)$ holds.

To measure the extent of deterioration in the quality of Nash packing due to
the effect of selfish and uncoordinated behavior of the players (items) in the
worst-case we use the Price of Anarchy (\textit{PoA}) and the Price of
Stability (\textit{PoS}). These are the standard measures of the quality of
the equilibria reached in uncoordinated selfish setting \cite{KP99,
  Rough}. The $PoA/ PoS$ of an instance $G$ of the Parametric Bin Packing game
are defined to be the ratio between the social cost of the worst/best Nash
equilibrium and the social optimum, respectively.  As packing problems are
usually studied via asymptotic measures, we consider asymptotic \emph{PoA} and
\emph{PoS} of the Parametric Bin Packing game \textit{BP}$(\alpha)$, that are
defined by taking a supremum over the \emph{PoA} and \emph{PoS} of all
instances of the Parametric Bin Packing game, for large sets $N$.

Recent research \cite{AFM07, FKLO07} initiated a study of measures
that separate the effect of the lack of coordination between players from the
effect of their selfishness. The measures considered are the Strong Price of
Anarchy (\textit{SPoA}) and the Strong Price of Stability
(\textit{SPoS}). These measures are defined similarly to the \emph{PoA} and
the \emph{PoS}, but only strong equilibria are considered.

These measures are well defined only when the sets $NE(G)$ and $SNE(G)$ are
not empty for any $G\in BP(\alpha)$.  Even though pure Nash equilibria are no
guaranteed to exist for general games, they always exist for the Bin Packing
game: The existence of pure Nash equilibria was proved in \cite{Bilo06} and
the existence of strong Nash equilibria was proved in \cite{EpsteinK08}.

As we study the $SPoA/SPoS$ measures in terms of the worst-case approximation
ratio of a greedy algorithm for bin packing, we define here the
\textit{parametric worst-case ratio} $R^\infty_A(\alpha)$ of algorithm $A$ by
\[ R^\infty_A(\alpha)=\lim_{k\rightarrow \infty} \sup_{I\in V_{\alpha}}
\bigg\langle \frac{A(I)}{OPT(I)}\ \bigg|\ OPT(I)=k \bigg\rangle,\]
where $A(I)$ denotes the number of bins used by algorithm $A$ to pack the set
$I$, $OPT(I)$ denotes the number of bins used in the optimal packing of $I$
and $V_{\alpha}$ is the set of all list $I$ for which the maximum size of the
items is bounded from above by $\alpha$. In this paper we use an equivalent
definition, where $R^\infty_A(\alpha)$ is defined as the smallest number such
that there exists a constant $K\geq 0$ for which $\displaystyle A(I)\leq
R^\infty_A(\alpha)\cdot OPT(I)+K$, for every list $I\in V_{\alpha}$.


\noindent{\bf Related work.} The first problems that were studied from game
theoretic point of view were job scheduling \cite{KP99, CV07, MS01} and
routing \cite{Rough, RoughgardenT00} problems. Since then, many other problems
have been considered in this setting.

The classic bin packing problem was introduced in the early 70's
\cite{Ullman71,JohnsonDUGG74}.  This problem and its variants are often met in
various real-life applications, and it has a special place in theoretical
computer science, as one of the first problems to which approximation
algorithms were suggested and analyzed with comparison to the optimal
algorithm. Bil\`{o} \cite{Bilo06} was the first to study the Bin Packing
problem from a game theoretic perspective. He proved that the Bin Packing game
admits a pure Nash equilibrium and provided non-tight bounds on the Price of
Anarchy. He also proved that the bin packing game converges to a pure Nash
equilibrium in a finite sequence of selfish improving steps, starting from any
initial configuration of the items; however, the number of steps may be
exponential. The quality of pure equilibria was further investigated by
Epstein and Kleiman \cite{EpsteinK08}. They proved that the Price of Stability
of the Bin Packing game equals to 1, and showed almost tight bounds for the
\textit{PoA}; namely, an upper bound of 1.6428 and a lower bound of
1.6416. Interestingly, this implies that the Price of Anarchy is not equal to
the approximation ratio of any natural algorithm for bin packing. Yu and Zhang
\cite{YuZ08} later designed a polynomial time algorithm to compute a packing
that is a pure Nash equilibrium. Finally, the \textit{SPoA} was analyzed in
\cite{EpsteinK08}.

A natural algorithm for the Bin Packing problem is the \textit{Subset Sum}
algorithm (or SS algorithm for short). In each iteration, the algorithm finds
among the unpacked items, a maximum size set of items that fits into a new
bin. The first mention of the \textit{Subset Sum} algorithm in the
literature is by Graham \cite{Graham72} who showed that its worst-case
approximation ratio $R^\infty_{SS}$ is at least $\sum_{i=1}^\infty
\frac{1}{2^i - 1} \approx 1.6067$. He also conjectured that this was indeed
the true approximation ratio of this algorithm. The SS algorithm can be
regarded as a refinement of the First-Fit algorithm \cite{JohnsonDUGG74},
whose approximation ratio is known to be $1.7$. Caprara and Pferschy
\cite{CapraraP04} gave the first non-trivial bound on the worst-case
performance of the SS algorithm, by showing that $R_{SS}^\infty(1)$ is at most
$\frac{4}{3}+\ln \frac {4}{3}\approx 1.6210$.  They also generalized their
results to the parametric case, giving lower and upper bounds on
$R_{SS}^\infty(\alpha)$ for $\alpha < 1$.


Surprisingly, the approximation ratio of the \textit{Subset Sum} is deeply
related to the Strong Price of Anarchy of the Bin Packing game. Indeed,
the two concepts are equivalent \cite{EpsteinK08}: Every output of the SS
algorithm is a strong Nash equilibrium, and every strong Nash equilibrium is
the output of some execution of the SS algorithm. Epstein and Kleiman
\cite{EpsteinK08} used this fact to show the existence of strong equilibria
for the Bin Packing game and to characterize the \textit{SPoA/SPoS} in terms
of this approximation ratio.

\noindent{\bf Our results.}
In this paper, we fully resolve the long standing open problem of finding the
exact approximation ratio of the Subset Sum algorithm, proving Graham's
conjecture to be true. This in turn implies a tight bound on the Strong Price
of Anarchy of the Bin-Packing game. Then we extend this result to the parametric
variant of bin packing where item sizes are all in an interval $(0,\alpha]$
for some $\alpha<1$. Interestingly, the ratio $R_{SS}^\infty(\alpha)$ lies
strictly between the upper and lower bounds of Caprara and Pferschy
\cite{CapraraP04} for all $\alpha \leq \frac12$. Finally, we study the pure
Price of Anarchy for the parametric variant and show nearly tight upper bounds
and lower bounds on it for any $\alpha<1$. The tight bound of 1 on the Price
of Stability proved in \cite{EpsteinK08} for the general unrestricted Bin
Packing game trivially carries over to the parametric case.

The main analytical tool we use to derive the claimed upper bounds is
\textit{weighting functions}---a technique widely used for the analysis of
algorithms for various packing problems \cite{Ullman71, JohnsonDUGG74,
  LeeLee85} and other greedy heuristics
\cite{conf/stacs/ImmorlicaMM05,journals/jacm/JainMMSV03}. The idea of such
weights is simple. Each item receives a weight according to its size
and its assignment in some fixed \textit{NE} packing. The weights are assigned
in a way that the cost of the packing (the number of the bins used) is close
to the total sum of weights. In order to complete the analysis, it is usually
necessary to bound the total weight that can be packed into a single bin of an
optimal solution.

Due to lack of space some of our proofs appear in the Appendix.

\section{Tight worst-case analysis of the Subset Sum algorithm}

In this section we prove tight bounds for the worst-case performance ratio of
the \textit{Subset Sum} (SS) algorithm for any $\alpha$. It was proved in
\cite{EpsteinK08} that the strong equilibria coincide with the packings
produced by the SS algorithm for Bin Packing.  The equivalence for the
\textit{SPoA}, \textit{SPoS} and the worst-case performance ratio of the
\textit{Subset Sum} algorithm which was also proved in \cite{EpsteinK08} still
applies for the Parametric Bin Packing game; indeed, it holds for all possible
lists of items (players), and in particular to lists where all items have size
at most $\alpha$. This allows us to characterize the \textit{SPoA/SPoS} in
terms of $R_{SS}^\infty(\alpha)$.

First we focus on the unrestricted case, that is, $\alpha =1$. Let
$\packing{B}{I}$ be the set of bins used by our algorithm and $\packing{O}{I}$
be the optimal packing for some instance $I$. We are interested in the
asymptotic worst-case performance of SS; namely, we want to identify constants
$\ratio{SS}$ and $\mar{SS}$ such that \begin{equation}
   \label{eq:goal}
   |\packing{B}{I}| \leq \ratio{SS}\, |\packing{O}{I}| + \mar{SS}.
\end{equation}

Using the weighting functions technique, we charge the ``cost" of the packing
to individual items and then show for each bin in $\packing{O}{I}$ that the
overall charge (weight) to items in the bin is not larger than $\ratio{SS}$.

Let $B \subseteq I$ be a bin in $\packing{B}{I}$. We use the
following short-hand notation $s(B) = \sum_{j \in B} s_j$ and
$\min(B) = \min_{j \in B} s_j$. Let $s_{\min}$ be the size of the
smallest yet-unpacked item just before opening $B$. For every $i
\in B$ we will charge item $i$ a share $w_i$ of the cost of
opening the bin, where
\begin{equation}
   \label{eq:share}
   w_i =
   \begin{cases}
     \frac{s_i}{s(B)} & \text{if } 1-s_{\min} \leq s(B), \\
     \ s_i & \text{otherwise.}
   \end{cases}
\end{equation}
These weights are very much related to the payments of selfish players (items)
in the Bin Packing game.

Let $w(B)$ denote the total weight of items in a bin $B$.
Note that if the size of items packed in $B$ is large enough
($s(B) \geq 1-s_{\min}$) then $w(B) =1$ and thus the charged
amount is enough to pay for $B$. Otherwise the charged amount only
pays for a $s(B)$ faction of the cost. Let $\hat{B}_1, \ldots,
\hat{B}_r$ be the bins that are underpaid listed in the order they
are opened by the algorithm and let $s_{\min}^i$ be the smallest
item available when $\hat{B}_i$ was opened. Notice that
$s_{\min}^i$ must belong to $\hat{B}_i$ otherwise we could safely
add the item to the bin. Also note that we cannot add
$s^{i+1}_{\min}$ to $s(\hat{B}_i)$, so we get
\begin{equation*}
s(\hat{B}_i) + s^{i+1}_{\min} > 1
  \Longrightarrow
  s^{i+1}_{\min} > s^i_{\min}.
\end{equation*}
Therefore, because of the definition of the SS heuristic, for all $i< r$, it must be case that swapping $s^i_{\min}$ with $s^{i+1}_{\min}$ in
$\hat{B}_i$ must yield a set that cannot be packed into a single bin, so we get \begin{equation*}
s(\hat{B}_i) - s^i_{\min} + s^{i+1}_{\min} > 1
  \Longrightarrow
  1-s(\hat{B}_i) < s^{i+1}_{\min} - s^i_{\min}.
 \end{equation*}

The total amount that is underpaid by all the $\hat{B}_i$ bins can be bounded as follows
\begin{equation*}
   \sum_{i=1}^r (1 - s(\hat{B}_i)) \leq \sum_{i=1}^{r-1} (s^{i+1}_{\min} -
   s^i_{\min}) + (1- s_{\min}^r) \leq 1.
\end{equation*} This amount will be absorbed by the additive constant term $\mar{SS}$ in our asymptotic bound
\eqref{eq:goal}.

Let $O$ be a set of items that can fit in a single bin, that is $s(O) \leq 1$, and denote with $s_1, s_2, \ldots, s_r$
the items contained in $O$, listed in \emph{reverse order} of how our algorithm packs them. Our goal is to show that
$\sum_{i \in O} w_i$ is not too big. To that end, we first establish some properties that these values must have
and then set up a mathematical program to find the sizes $s_1, \ldots, s_r$ obeying these properties and maximizing
$w(O)$. Consider the point in time when our algorithm packs $s_i$. Let $B$ be the bin the algorithm uses to pack
$s_i$ and let $O_i = \{1, \ldots, i\}$.

Because $O_i$ is a candidate bin for our algorithm we get $s(B) \geq s(O_i)$. Therefore, by \eqref{eq:share}, we have
\begin{equation}
   \label{eq:alpha-constraint-pack-O}
   w_i \leq \frac{s_i}{s(O_i)}.
\end{equation}
Notice that if $s(B) < 1 - \min(O_i)$ then $i$'s share is $s_i$.  Therefore, we always have \begin{equation}
   \label{eq:alpha-constraint-min-O}
   w_i \leq \frac{s_i}{1-\min(O_i)}.
\end{equation}

Our job now is to find sizes $s_1, \ldots, s_r$ maximizing $w(O)$ subject to \eqref{eq:alpha-constraint-pack-O}
and \eqref{eq:alpha-constraint-min-O}. Equivalently, we are to determine the value of the following mathematical
program

\begin{gather} \label{LP:factor} \tag{$\mathrm{MP}_r$}
  \text{maximize} \ \sum_{i=1}^r \frac{s_i}{\max \left\{ \sum_{j=1}^i s_j, 1-
      \min_{1\leq j \leq i} s_j \right\} } \\[-6ex] \notag
\end{gather} \hspace{1.5cm} subject to \\[-3ex] \begin{align}
  \sum_{i=1}^r s_i\ & \leq 1 \notag   \\
  \notag
  s_i  & \geq 0 & \forall\ i \in [r]
\end{align}

Let $\lambda_r$ be the value of \eqref{LP:factor} and let $\lambda = \sup_r \lambda_r$. The following theorem shows
that the worst-case approximation ratio of the SS algorithm is precisely $\lambda$.

\begin{theorem}
  \label{thm:worst-case}
  For every instance $I$, we have $|\packing{B}{I}| \leq \lambda\,
  |\packing{O}{I}| + 1$. Furthermore, for every~$\delta > 0$, there exists an
  instance $I$ such that $|\packing{B}{I}| \geq \left( \lambda - \delta
  \right) \, |\packing{O}{I}|$.
\end{theorem}

The necessary tools for proving the upper bound have been laid out above, we just need to put everything together:
\[\abs{\packing{B}{I}} \leq \sum_{B \in \packing{B}{I}} \sum_{i \in B} w_i + 1 = \sum_{O \in \packing{O}{I}}
\sum_{i \in O} w_i + 1 \leq \sum_{O \in
  \packing{O}{I}} \lambda_{\abs{O}} + 1 \leq \lambda \abs{\packing{O}{I}} + 1. \]

To be able to prove the claimed lower bound, we first need to study some
properties of \eqref{LP:factor}. The following lemma fully characterizes the
optimal solutions of \eqref{LP:factor}.

\begin{lemma}
  \label{lem:optimal-solution}
  The optimal solution to \eqref{LP:factor} is
  \begin{equation*}
     s^*_i =
     \begin{cases}
       2^{-i} & \text{if } i < r, \\
       2^{-r+1} & \text{if } i = r.
   \end{cases}
  \end{equation*}
\end{lemma}

It follows that the optimal value of \eqref{LP:factor} is
\(\lambda_r = \sum_{i =1}^{r-1} \frac{1}{2^i-1} + \frac{1}{2^{r-1}}\).
This expression increases as $r$ grows. Therefore, the value is
always at most \[\lambda = \sum_{i=1}^\infty \frac{1}{2^i -1}.\]

To lower bound the performance of the SS algorithm we use a
construction based on Graham's original paper: The instance $I$
has for each $i \in [r-1]$, $N$ items of size $2^{-i} +
\varepsilon $, and for $i=r$, $N$ items of size
$2^{-r+1}-r\varepsilon$, where $\varepsilon = 2^{-2r}$ and $N$
is large enough so that $N/s_i$ is integral for all $i$. The SS
algorithm first packs the smallest items into $N/2^{r-1}$ bins,
then it packs the next smallest items into $N/(2^{r-1} -1)$
bins, the next items into $N/(2^{r-2} -1)$ bins, and so on. On
the other hand, the optimal solution uses only $N$ bins. If we
choose $r$ to be such that $2^{r}-1 \geq \delta^{-1}$ then we
get \[\abs{\packing{B}{I}} = \lambda_r \abs{\packing{O}{I}}
\geq \left( \lambda -
  \frac{1}{2^{r} -1} \right) \abs{\packing{O}{I}} \geq (\lambda -\delta)
\abs{\packing{O}{I}}.\]
Note that this lower bound example, for the case where there are $r$ distinct
item sizes, gives exactly the upper bound we found for $\mathrm{MP}_r$.

\begin{corollary}
  For $\alpha \in (\frac12,1]$, the approximation ratio of the SS algorithm is
  $R_{SS}^\infty(\alpha) = \sum_{i =1}^{\infty} \frac{1}{2^i -1} \approx
  1.6067$. Furthermore, the SPoA/SPoS of the $BP(\alpha)$ game has the same value.
\end{corollary}


\noindent{\bf Parametric case.} To get a better picture of the performance of
SS, we generalize Theorem~\ref{thm:worst-case} to instances where the size of
the largest item is bounded by a parameter $\alpha$. Our goal is to establish
the worst-case performance of the SS algorithm for instance in $V_\alpha$ for
all $\alpha < 1$.

Let $t$ be the smallest integer such that $\alpha \leq \frac1t$. We proceed as
we did before but with a slightly different weighting function: \begin{equation}
   \label{eq:share-param}
   w_i =
   \begin{cases}
     \frac{s_i}{s(B)} & \text{if } \max\,\{1-s_{\min}, \frac{t}{t+1}\} \leq s(B), \\
     \ s_i & \text{otherwise.}
   \end{cases}
\end{equation}

As before there will be some bins that are underpaid. Let $\hat{B}_1, \ldots,
\hat{B}_r$ be these bins and let $s^i_{\min}$ be smallest yet-unpacked item
when the algorithm opened $\hat{B}_i$. These bins only pay for a
$s(\hat{B}_i)$ fraction of their cost. Even though we now have a more
restrictive charging rule, the total amount underpaid is still at most~1. For
all $i< r$, when $s(\hat{B}_i) < 1-s^i_{\min}$, the same argument used above
yields \[1-s(\hat{B}_i) < s^{i+1}_{\min} - s^i_{\min}.\]

Suppose that for some $i$ we have $s(\hat{B}_i) < \frac{t}{t+1}$ but
$s(\hat{B}_i) > 1 - s^i_{\min}$. Note that this implies $s^i_{\min} >
1/(t+1)$. Since at this point every item has size in $(\frac1{t+1},\frac1t]$,
if there were left at least $t$ items left just before $\hat{B}_i$ was opened,
we could pack a bin with total size greater than $\frac{t}{t+1}$. Therefore,
$\hat{B}_i$ must be the last bin packed by the algorithm. Regardless whether
such a bin exists or not, we always have $1- s(\hat{B}_r) \leq 1 -
s^r_{\min}$. Hence, the total amount underpaid is
\begin{equation*}
   \sum_{i=1}^r 1 - s(\hat{B}_i) \leq \sum_{i=1}^{r-1} (s^{i+1}_{\min} -
   s^i_{\min}) + (1- s_{\min}^r) \leq 1.
\end{equation*}

The new weighting function \eqref{eq:share-param} leads to the following
mathematical program

\begin{gather} \label{LP:factor-param} \tag{$\mathrm{MP}_r^t$}
  \text{maximize} \hspace{1em} \sum_{i=1}^r \frac{s_i}{\max \left\{
      \sum_{j=0}^i s_j,\, 1- \min_{1\leq j \leq i} s_j,\, t/(t+1)\right\} } \\[-6ex] \notag
\end{gather} \hspace{1.5cm} subject to \\[-3ex] \begin{align}
  \sum_{i=0}^r s_i\ & \leq 1 \notag   \\
  \notag
  s_i  & \geq 0 & \forall\ i \in [r] \\
  \notag
  s_i  & \leq 1/t & \forall\ i \in [r-1]
\end{align}

Notice that $s_r$ is allowed to be greater than $1/t$. This relaxation does
not affect the value of the optimal solution, but it helps to simplify our
analysis. From now on, we assume that $r\geq t$; for otherwise the program
become trivial. Define $\lambda_r^t$ to be the value of
\eqref{LP:factor-param} and $\lambda^t = \sup_r \lambda_r^t$.

\begin{theorem}
  \label{thm:worst-case-param}
  Let $t \geq 2$ be an integer and $\alpha \in (\frac{1}{t+1}, \frac1m]$. For
  every instance $I \in V_\alpha$, we have $|\packing{B}{I}| \leq \lambda^t\,
  |\packing{O}{I}| + 1$. Furthermore, for every $\delta > 0$, there exist an
  instance $I \in V_\alpha$ such that $|\packing{B}{I}| \geq \left( \lambda^t - \delta
  \right) \, |\packing{O}{I}|$.
\end{theorem}

The proof of the upper bound is identical to that of
Theorem~\ref{thm:worst-case}. We only need to derive the counterpart of
Lemma~\ref{lem:optimal-solution} for \eqref{LP:factor-param}. Unlike its
predecessor, Lemma~\ref{lem:optimal-solution-param} does not fully
characterize the structure of the optimal solution of
\eqref{LP:factor-param}. Rather, we define an optimal solution $s^*$ as a
function of a free parameter $x$.

\begin{lemma}
  \label{lem:optimal-solution-param}
  An optimal solution to \eqref{LP:factor-param} has the form
  \begin{equation*}
     s^*_i =
     \begin{cases}
       x & \text{if } i < t, \\
       \frac{1-x(t-1)}{2^{i-t + 1}} & \text{if } t \leq i < r, \\
       \frac{1-x (t-1)}{2^{r-t}} & \text{if } i = r,
   \end{cases}
  \end{equation*}
  for some $x \in [\frac{1}{t+1}, \frac{1}{t}]$.
\end{lemma}

For any $x \in [\frac1{t+1}, \frac1t]$, we can construct a
solution $s^*$ for \eqref{LP:factor-param} as described in
Lemma~\ref{lem:optimal-solution-param}. Let $\lambda_r^t(x)$ be
the value of the value of this solution, that is,
\[\lambda^t_r(x) = x\, (t-1)\, \frac{t+1}{t} +
\sum_{i = 1}^{r-t} \frac{1}{\frac{2^i}{1-(t-1)x}-1}
+\frac{1}{\frac{2^{r-t}}{1-(t-1)x}}.\]
For any fixed $x$, the quantity $\lambda_r^t(x)$ increases as $r \rightarrow
\infty$. Therefore, it is enough to look at its limit value, which we denote
by $\lambda^t(x)$:
\[\lambda^t(x) = \lim_{r \rightarrow \infty} \lambda^t_r(x) = x\, (t-1)\,
\frac{t+1}{t} + \sum_{i = 1}^{\infty} \frac{1}{\frac{2^i}{1-(t-1)x}-1}.\]
It only remains to identify the value $x \in [\frac{1}{t+1}, \frac1t]$
maximizing $\lambda^t(x)$.

\begin{lemma}
  \label{lem:maximum-x}
  For every $t \geq 2$, the function $\lambda^t(x)$ in the domain
  $[\frac{1}{t+1}, \frac1t]$ attains its maximum at $x = \frac{1}{t+1}$.
\end{lemma}

It follows that $\lambda^t = \lambda^t\big(\frac{1}{t+1}\big)$,
that is, \[\lambda^t = 1 + \sum_{i = 1}^\infty \frac{1}{(t+1) \,
2^i -1}.\]

Note that for a specific value of $r$,
\[\lambda^t_r\big({\textstyle \frac{1}{t+1}}\big) = 1+ \sum_{i
= 1}^{r-t-1} \frac{1}{(t+1) \,2^{i}-1} +\frac{1}{(t+1) \, 2^{r-t-1}}.\]

For the lower bound on the performance of the SS algorithm, consider the
instance $I$ that for each $i \in [t]$ has $N$ items of size $\frac{1}{t+1} +
\varepsilon $, for each $i \in (t,r)$, it has $N$ items of size
$\frac{1}{(t+1)\,2^{i-t}} + \varepsilon$, and for $i=r$, there are $N$ items
of size $\frac{1}{(t+1)\,2^{r-1-t}} -r \varepsilon$, where $\varepsilon =
\frac{1}{(t+1)^2 \, 2^{-2r}}$ and $N$ is large enough so that $N/s_i$ is
integral for all $i$. The SS algorithm first packs the smallest items into
$\frac{N}{(t+1) 2^{r-t-1}}$ bins, then it packs the next smallest items into
$\frac{N}{(t+1) 2^{r-1-t} -1}$ bins, and so on until reaching the items of
size $\frac{1}{t+1}+\varepsilon$ which are packed into $N$ bins. The optimal
solution uses $N$ bins. If we choose $r$ to be such that $(t+1)2^{r-t}-1 \geq
\delta^{-1}$ then we get



\[ \abs{\packing{B}{I}} = \lambda_r^t\big({ \textstyle \frac{1}{t+1} } \big)\,
  \abs{\packing{O}{I}}\ \geq \left( \lambda^t - \frac{1}{(t\!+\!1) 2^{r-t} -1}
  \right) \abs{\packing{O}{I}}\geq \big(\lambda^t -\delta\big)
  \abs{\packing{O}{I}}.\]

\begin{corollary}
  For each integer $t \geq 1$ and $\alpha \in (\frac{1}{t+1}, \frac{1}{t}]$,
  the SS algorithm has an approximation ratio of $R_{SS}^\infty(\alpha) = 1 +
  \sum_{i =1}^{\infty} \frac{1}{(t+1) 2^i -1}$. Furthermore, the SPoA/SPoS of the
  $BP(\alpha)$ game has the same value.
\end{corollary}

Figure~\ref{fig:ratio-comparison} compares our bound with the previously known upper
bounds and lower bounds of Caprara and Pferschy \cite{CapraraP04}. Note that
the true ratio lies strictly between previous bounds.

\section{Analysis of the Price of Anarchy}

We now provide a lower bound for the Price of Anarchy of the parametric bin packing game with bounded size items. In addition we prove a very close upper bound for each value of $\frac{1}{t+1}<\alpha\leq\frac1 t$ for a positive integer $t\geq 2$, that is, for all $0<\alpha \leq \frac 1 2$. The case $\frac 1 2 <\alpha <1$ ($t=1$) was extensively studied in \cite{EpsteinK08}.


\noindent{\bf A construction of lower bound on the \textit{PoA} of parametric
  Bin Packing.} In this section we give the construction of a lower bound on
\textit{PoA$(\alpha)$}. For each value of $t\geq 2$ we present a set of items
which consists of multiple item lists. This construction is somewhat related
to the construction we gave in \cite{EpsteinK08} for $\frac 1 2<\alpha \leq
1$, though it is not a generalization of the former, which strongly relies on
the fact that each item of size larger than $\frac 1 2$ can be packed alone in
a bin of the $NE$ solution, whereas in the parametric case there are no such
items. It is based upon techniques that are often used to design lower bounds
on bin packing algorithms (see e.g., \cite{LeeLee85}). We should note that our
construction differs from these constructions in the notion of order in which
packed bins are created (which does not exist here) and the demand that each
bin satisfies the Nash stability property. Our lower bound is given by the
following theorem, whose proof appears in Appendix~\ref{app:proof-thm-20}.
\begin{theorem} \label{thm:20} For each  integer $t\geq 2$ and $\alpha \in
  (\frac{1}{t+1}, \frac{1}{t}]$, the \textit{PoA} of the $BP(\alpha)$ game is
  at least \\ \(\frac{t^2+\sum\limits_{j=1}^{\infty} {(t+1)^{-j}\cdot
      2^{-j(j-1)/2}}}{t(t-1)+1}.\)
\end{theorem}

\noindent{\bf An upper bound on the \textit{PoA} of parametric
Bin Packing.} We now provide a close upper bound on
\textit{PoA}($\alpha$) for a positive integer $t\geq 2$. The
technique used in \cite{EpsteinK08} can be considered as a
refinement of the one we use here, and here we are also
required to use additional combinatorial propertiies of the NE
packing. To bound the \textit{PoA} from above, we prove the
following theorem. \begin{theorem} \label{thm:36} For each
integer $t\geq 2$, for any instance of the parametric bin
packing game $G\in BP(\frac{1}{t})$: Any \textit{NE} packing
uses at most $\big(\frac{2t^3+t^2+2}{(2t+1)(t^2-t+1)}\big)\cdot
OPT(G)+5$ bins, where $OPT(G)$ is the number of bins used in a
coordinated optimal packing. \end{theorem} \begin{proof} Let us
consider a packing $b$ of the items in $N_G$ which admits
\textit{NE} conditions. We classify the bins according to their
loads into four groups-$\mathcal{A}$,
$\mathcal{B}$,$\mathcal{C}$ and $\mathcal{D}$. The cases $t=2$
and $t\geq 3$ are treated separately. For $t=2$: group
$\mathcal{A}$- contains bins with loads of more than
$\frac{5}{6}$; Group $\mathcal{B}$- contains bins with loads in
$(\frac{3}{4}, \frac{5}{6}]$; Group $\mathcal{C}$- contains
bins with loads in $(\frac{17}{24}, \frac{3}{4}]$; Group
$\mathcal{D}$- contains bins with loads not greater than
$\frac{17}{24}$. For $t\geq 3$: group $\mathcal{A}$- contains
bins with loads of more than $\frac{2t+1}{2(t+1)}$; Group
$\mathcal{B}$- contains bins with loads in  $(\frac{t+1}{t+2},
\frac{2t+1}{2(t+1)}]$; Group $\mathcal{C}$- contains bins with
loads in $(\frac{t^2-t+1}{t^2}, \frac{t+1}{t+2}]$; Group
$\mathcal{D}$- contains bins with loads not greater than
$\frac{t^2-t+1}{t^2}$. This partition is well defined, as
$\frac{t}{t+1}<\frac{t^2-t+1}{t^2}$,
$\frac{t^2-t+1}{t^2}<\frac{t+1}{t+2}$ and
$\frac{t+1}{t+2}<\frac{2t+1}{2(t+1)}$ for any $t\geq 3$. We
denote the cardinality of these groups by $n_\mathcal{A},
n_\mathcal{B}, n_\mathcal{C}$ and $n_\mathcal{D}$,
respectively. Hence,
$NE=n_\mathcal{A}+n_\mathcal{B}+n_\mathcal{C}+n_\mathcal{D}$.
We list the bins in each group from left to right in
non-increasing order w.r.t. their loads. Our purpose is to find
an upper bound on the total number of bins in these four
groups.

In the case $n_\mathcal{D}<3$, using the fact that $OPT\geq \sum_{i=1}^{n}s_i$
we consider two sub-cases:
\begin{itemize}
  \item[$\bullet$] For $t=2$, this means that all bins in packing $b$ (except for at most 2) have load of at least $\frac{17}{24}$, thus $OPT\geq \frac{17}{24} NE$, and $PoA\leq \frac{24}{17}<\frac{22}{15}$.

  \item[$\bullet$] For $t\geq 3$, this means that all bins in packing $b$ (except for at most 2) have load of at least $\frac{t^2-t+1}{t^2}$, thus $OPT\geq \frac{t^2-t+1}{t^2} NE$, and $PoA\leq \frac {t^2}{t^2-t+1}<\frac{2t^3+t^2+2}{(2t+1)(t^2-t+1)}$.
\end{itemize}

In the rest of the analysis we assume that $n_\mathcal{D}\geq 3$. We start with a simple lower bound on the load of the bins (except possibly at most two bins) in a \textit{NE} packing.

\begin{claim} \label{thm:40}
For a positive integer $t\geq2$, all the bins in NE packing $b$ (except for maybe a constant number of bins) are at least $\frac{t}{t+1}$ full.
\end{claim}
Moreover, the fact that any \textit{NE} packing can be produced by a run of FF actually implies that the worst-case asymptotic ratio of FF, which is known to be $\frac{t+1}{t}$ for $t\geq 2$, upper-bounds the \textit{PoA}. But, as we show further, the upper-bound we provide on the \textit{PoA} is tighter than this trivial bound for any $t\geq 2$.

From Claim \ref{thm:40} it is evident that all the bins (except for maybe two) in group $\mathcal{D}$ have loads in $(\frac{2}{3}, \frac{17}{24}]$ for $t=2$, or in $(\frac{t}{t+1}, \frac{t^2-t+1}{t^2}]$ for $t\geq 3$.

\begin{claim} \label{thm:41}
For a positive integer $t\geq2$, in a \textit{NE} packing $b$, all bins that are filled by less than $\frac{2t+1}{2(t+1)}$ (i.e. bins in groups $\mathcal{B}$, $\mathcal{C}$ and $\mathcal{D}$), except for maybe a constant number of bins, contain exactly $t$ items with sizes in $(\frac{t-1}{t^2}, \frac{1}{t}]$.
\end{claim}

Henceforth, we call the bins in groups $\mathcal{B}$, $\mathcal{C}$ and $\mathcal{D}$ that contain exactly $t$ items with sizes in $(\frac{t-1}{t^2}, \frac{1}{t}]$ for $t\geq 3$, or exactly $2$ items of sizes in $(\frac{7}{24}, \frac{1}{2}]$ for $t=2$ \textit{regular} bins, and refer to each one of those items as $t$-item.

To derive the upper bound on the total number of bins in the \textit{NE} packing $b$, we use the \textit{weighting functions} technique.

We define for each value of $t\geq 2 $ a weighting function $w_t$ on the items, in the following manner. The weight $w_t(x)$ of an item of size $x$ which is packed in a bin of group $\mathcal{A}$ in a packing $b$ is: $w_t(x)=\frac{2(t+1)}{2t+1}x$. The weight $w_t(x)$ of an item of size $x$ which is packed in a regular bin of load $L<\frac{2t+1}{2(t+1)}$ in a packing $b$ is: $w_t(x)=\frac{2(t+1)}{2t+1}x+ \frac{(1-\frac{2(t+1)}{2t+1} L)}{k}$, where $k$ is the number of items in the bin of $x$. The purpose of the addition term  $\frac{(1-\frac{2(t+1)}{2t+1} L)}{k}$ is to complete the weight of any bin in the packing to $1$. Clearly, any bin in group $\mathcal{A}$ (which is full by more than $\frac{2t+1}{2(t+1)}$) will have a total weight of at least 1. Any of the less filled bins from groups $\mathcal{B}$, $\mathcal{C}$ and $\mathcal{D}$ will have a weight of 1 as $\frac{2(t+1)}{2t+1}\cdot L+ \frac{(1-\frac{2(t+1)}{2t+1} L)}{t}\cdot t=1$, and each of the $t$ items packed in each one of these bins (except maybe 5 bins) will get an addition of at most $\frac{1-\frac{2(t+1)}{2t+1}\cdot \frac{t}{t+1}}{t}=\frac{1}{t(2t+1)}$.

For the 5 special bins, the first weighting function applies, and the weight of each bin is non-negative.


Now, we need to bound from above the weight observed by a bin in the optimal packing of these items. First, note that in a bin of the optimal packing for $t\geq 2$ there can be at most $t+1$ $t$-items from the regular bins of groups $\mathcal{B}$, $\mathcal{C}$ and $\mathcal{D}$. For $t=2$ the size of these items is greater than $\frac{7}{24}$, and the size of four of these items exceeds 1. For $t\geq 3$ the size of these items is greater than $\frac{t-1}{t^2}$, and the size of $t+2$ of these items, which is at least $(t+2)\cdot \frac{t-1}{t^2}=1+\frac{t-2}{t^2}$, also exceeds 1.

The weight of a bin in an optimal packing that has a load $S$ and contains $t+1$ $t$-items that come from bins of groups $\mathcal{B}$, $\mathcal{C}$ and $\mathcal{D}$ in $b$, is at most: \begin{displaymath} \frac{2(t+1)}{2t+1}\cdot S +(t+1)\cdot\frac{1}{t(2t+1)}\leq \frac{2(t+1)}{2t+1}+\frac{t+1}{t(2t+1)}=\frac{2t^2+3t+1}{t(2t+1)}=\frac{t+1}{t}.\end{displaymath}
The weight of a bin in an optimal packing that has a load $S$ and contains at most $t$ $t$-items that came from bins of groups $\mathcal{B}$, $\mathcal{C}$ and $\mathcal{D}$ in $b$, is at most: \begin{displaymath} \frac{2(t+1)}{2t+1}\cdot S +t\cdot\frac{1}{t(2t+1)}\leq \frac{2(t+1)}{2t+1}+\frac{t}{t(2t+1)}=\frac{2t^2+3t}{t(2t+1)}.\end{displaymath}

We claim that in any optimal packing, the fraction of the number of bins that contain $t+1$ $t$-items from bins of groups $\mathcal{B}$, $\mathcal{C}$ and $\mathcal{D}$ out of total number of bins is at most $\frac{t(t-1)}{t^2-t+1}$.

To establish this, we consider all the bins in the optimal packing that contain exactly $t+1$ $t$-items from groups $\mathcal{B}$, $\mathcal{C}$ and $\mathcal{D}$ (and maybe additional items as well), let the number of such bins be $N_t$.

If $N_t=0$, we are done as then the total weight of all the items in $N_G$ is at most $ W(N_G)\leq \big(\frac{2t+3}{2t+1}\big)\cdot OPT(G).$ 
As $n_\mathcal{A}+n_\mathcal{B}+n_\mathcal{C}+n_\mathcal{D} -5\leq W(N_G)$,  we get that $NE\leq \big(\frac{2t+3}{2t+1}\big)\cdot OPT(G)+5<\big(\frac{2t^3+t^2+2}{(2t+1)(t^2-t+1)}\big)\cdot OPT(G)+5$.
Else, we prove the following claim.
\begin{claim} \label{thm:50}
Among the $N_t\cdot (t+1)$ $t$-items that are packed in $(t+1)$-tuples in the  bins of the optimal packing, only at most $(N_t-1)\cdot t$ are packed together in  $t$-tuples in bins that belong to groups $\mathcal{B}$, $\mathcal{C}$ and $\mathcal{D}$ in the $NE$ packing.
\end{claim}

Hence, at most $(N_t-1)\cdot t$ $t$-items out of $N_t\cdot (t+1)$ are packed together in  $t$-tuples in bins from groups $\mathcal{B}$, $\mathcal{C}$ and $\mathcal{D}$ in the $NE$ packing $b$.
The remaining $N_t+t$ $t$-items are also packed in bins of groups $\mathcal{B}$, $\mathcal{C}$ and $\mathcal{D}$ in $b$, but they share their bin with at most $(t-2)$ other $t$-items from the $N_t$ bins from the optimal packing, and at least one $t$-item that is not packed in one of these $N_t$ bins. In total, there are at least $\frac{N_t+t}{t-1}$  $t$-items that are not packed in one of the $N_t$ bins in discussion, and they are packed with at most $t-1$ other such items in the optimal packing.

Thus, in the optimal packing for any $N_t$ bins with $t+1$ items of size in $(\frac{t-1}{t^2}, \frac 1 t]$  there are
at least $\frac{N_t+t}{t(t-1)}$ bins that have at most $t$ such items. Letting $N_t$ be very large in comparison to $t$ gives us the claimed proportions. We conclude that in average, the weight of any bin of the optimal packing is at most: \begin{displaymath}
\frac{t(t-1)\cdot \frac{t+1}{t}+\frac{2t+3}{2t+1}}{t(t-1)+1}=\frac{2t^3+t^2+2}{(2t+1)(t^2-t+1)}.
\end{displaymath}
Hence, the total weight of all the items in $N_G$ is at most $W(N_G)\leq \big(\frac{2t^3+t^2+2}{(2t+1)(t^2-t+1)}\big)\cdot OPT(G).$
As $n_\mathcal{A}+n_\mathcal{B}+n_\mathcal{C}+n_\mathcal{D} -5\leq W(N_G)$,  we get that $NE\leq \big(\frac{2t^3+t^2+2}{(2t+1)(t^2-t+1)}\big)\cdot OPT(G)+5$
\end{proof}

A more careful consideration of the contents of special bins allows to reduce the additive constant to 2.

\begin{theorem} \label{thm:37} For each integer $t \geq 2$ and $\alpha \in
  (\frac{1}{t+1}, \frac 1 t]$, the \textit{PoA} of the parametric bin packing
  game $BP(\alpha)$ is at most $\frac{2t^3+t^2+2}{(2t+1)(t^2-t+1)}$.
\end{theorem}
\begin{proof}
The asserted upper bound on the $PoA$ follows directly from Theorem \ref{thm:36}.
\end{proof}


\begin{figure}[t]
  \centering

\hspace{0.3cm}
  \subfigure[\label{fig:ratio-comparison} A comparison of our analysis of
  $R_{SS}^\infty(\alpha)$ with Caprara and Pferschy's (CP). The true ratio
  lies between the previously known upper and lower
  bounds.]{\includegraphics[angle=90,width=0.41\textwidth]{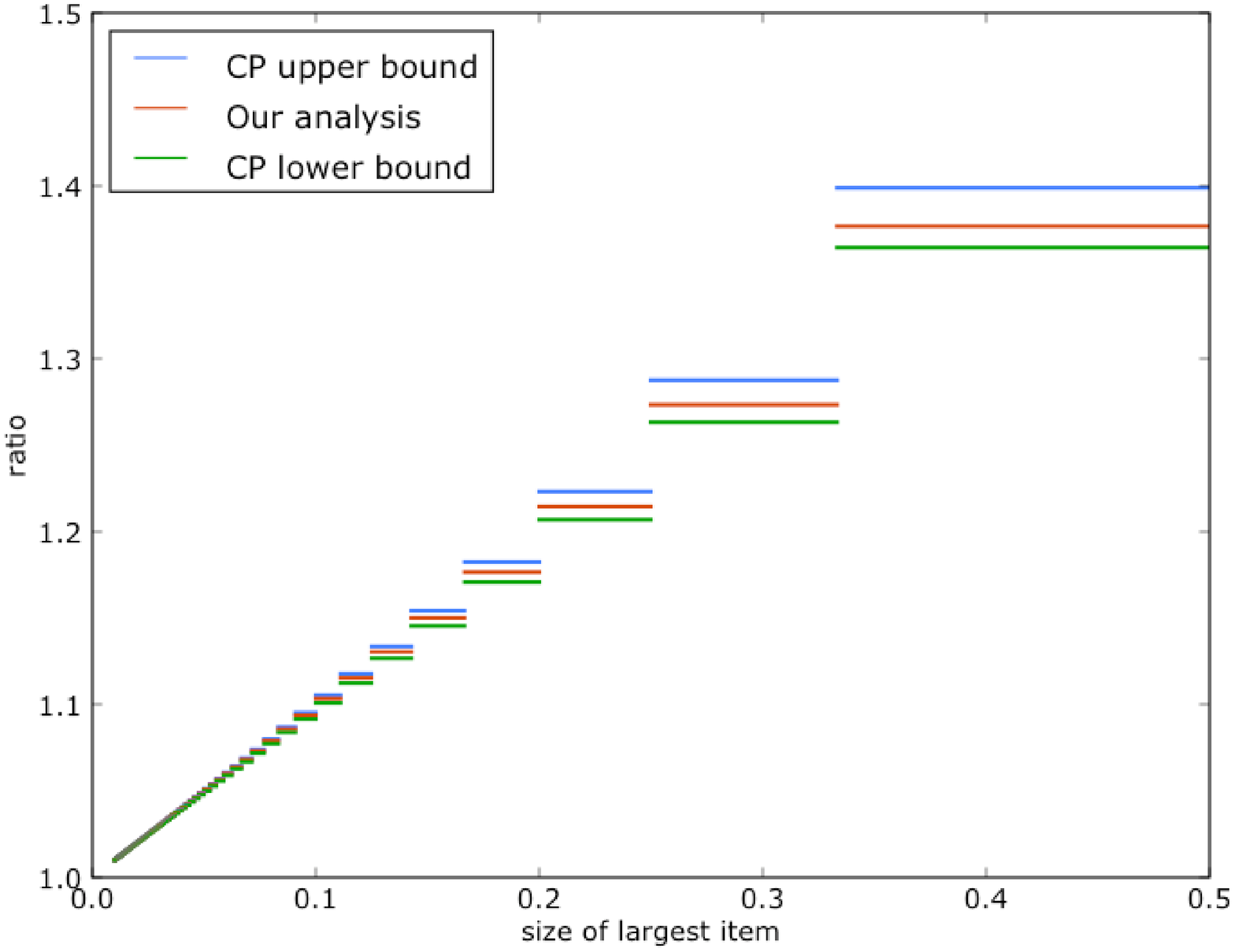}}
  \hspace{2em} \subfigure[\label{fig:poa-comparison} Almost matching upper and
  lower bounds for the PoA of the parametric bin packing
  game. ]{\includegraphics[angle=90,width=0.41\textwidth]{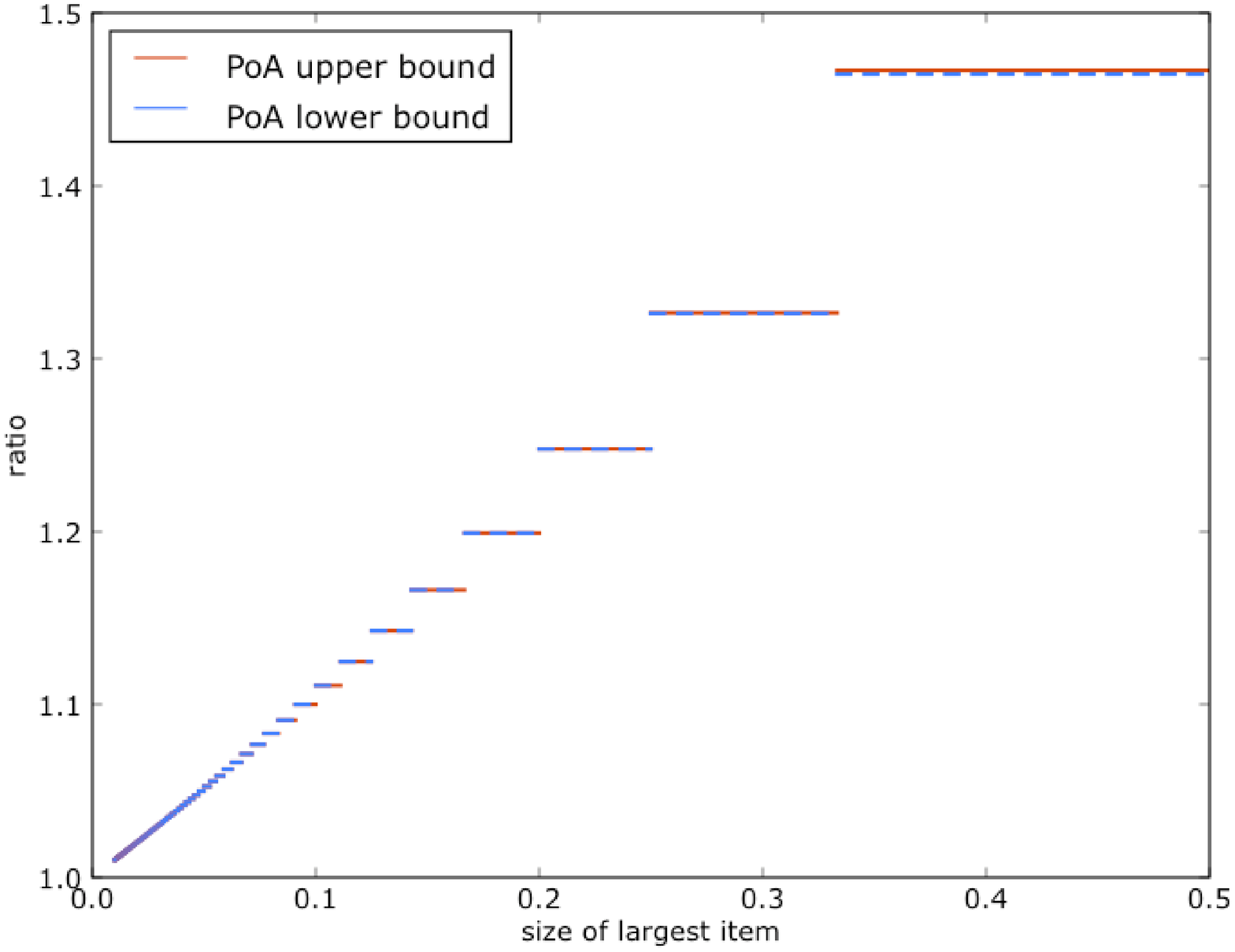}}

  \subtable[Comparison of the worst-case ratio of \textit{FFD}, \textit{SS},
  \textit{FF} and \textit{PoA} as a function of $\alpha$ when
  $\alpha\leq \frac{1}{t}$, for $t  = 1, \ldots, 10$.]{ \footnotesize
    \begin{tabular}{c||c|c|c|c|c|c}
      \hline
      &$R_{\mathit{FFD}}(\alpha)$ \cite{JohnsonDUGG74} & CP lb \cite{CapraraP04}&
      $R_{\mathit{SS}}(\alpha)$ & CP ub  \cite{CapraraP04}& $PoA(\alpha)$ & $R_{\mathit{FF}}(\alpha)$ \cite{JohnsonDUGG74}\\
      \hline \hline
      $t=1$ & 1.222222 & 1.606695 \cite{Graham72}& 1.606695 & 1.621015 & [1.641632, 1.642857] \cite{EpsteinK08}& 1.700000\\
      $t=2$ & 1.183333 & 1.364307 & 1.376643 & 1.398793 & [1.464571, 1.466667] & 1.500000 \\
      $t=3$ & 1.166667 & 1.263293 & 1.273361 & 1.287682 & [1.326180, 1.326530] & 1.333333 \\
      $t=4$ & 1.150000 & 1.206935 & 1.214594 & 1.223143 & [1.247771, 1.247863] & 1.250000 \\
      $t=5$ & 1.138095 & 1.170745 & 1.176643 & 1.182321 & [1.199102, 1.199134] & 1.200000 \\
      $t=6$ & 1.119048 & 1.145460 & 1.150106 & 1.154150 & [1.166239, 1.166253] & 1.166667 \\
      $t=7$ & 1.109127 & 1.126763 & 1.130504 & 1.133531 & [1.142629, 1.142635] & 1.142857 \\
      $t=8$ & 1.097222 & 1.112360 & 1.115433 & 1.117783 & [1.124867, 1.124871] & 1.125000 \\
      $t=9$ & 1.089899 & 1.100918 & 1.103483 & 1.105360 & [1.111029, 1.111031] & 1.111111 \\
      $t=10$ & 1.081818 & 1.091603 & 1.093776 & 1.095310 & [1.099946, 1.099947] & 1.100000 \\
      \hline
\end{tabular} \label{table:results}
}

   \caption{Our results at a glance.}

\end{figure}

\section{Concluding Remarks}

In order to illustrate the results in the paper, we report in Figure
\ref{table:results} the values for the worst-case ratio of the SS algorithm
for various values of $\alpha$ along with previously known upper and lower
bounds of Caprara and Pferschy \cite{CapraraP04}, and the worst-case
approximation ratios of FF and FFD algorithm Bin Packing. We also include the
range of possible values for the PoA for different values of
$\alpha$. Figure~\ref{fig:poa-comparison} shows our (almost matching) upper
and lower bound on the PoA.  We conjecture that the true value of the $PoA$
equals our lower bound from Theorem~\ref{thm:20}.

small
\bibliographystyle{abbrv}


\appendix
\section{Omitted proofs}
\subsection{Proof of Lemma \ref{lem:optimal-solution}}
Let $s$ be a solution to \eqref{LP:factor} other than $s^*$. The plan is to
show that $s$ is not optimal by improving its cost. First we argue that
without loss of generality $\sum_{i = 1}^r s_i = 1$. Indeed, if that was not
the case then consider the new solution
  \begin{equation*}
    s'_i =
    \begin{cases}
      s_i & \text{if } i < r, \\
      1- \sum_{j=1}^{r-1} s_j & \text{if } i = r.
    \end{cases}
  \end{equation*}
  The difference between the value of $s'$ and the value of $s$ comes from the
  $r$th term in the objective of \eqref{LP:factor}. Since $s'_r > s_r$, this
  difference is at least
  \[ s'_r - \frac{s_r}{1-s'_r + s_r} > s'_r - \frac{s'_r}{1-s'_r + s'_r} = 0.\]

  Let $i$ be the first index such that $s_i \neq s^*_i$. First, we consider
  the case $s_i < s^*_i$. Let $h > i$ be the smallest index such that
  $\sum_{j=1}^h s_j \geq 1- \min_{1\leq j \leq h} s_j$. (Note that such $h$
  must exist because the condition is satisfied by $r$ and $i<r$ by our
  assumption that $\sum_i s_i = 1$.)  We construct a new solution $s'$ from
  $s$ by slightly increasing $s_i$ and slightly decreasing $s_{h}$ by the same
  $\varepsilon$ amount (note that $s_h$ must be non-zero). We would like to
  argue that the overall change in value is positive. To that end, we examine
  how each term in the objective of \eqref{LP:factor} changes with the update.

  \begin{itemize}
    \renewcommand{\labelitemi}{$\triangleright$}
    \item
  For $k \in [1, i)$ the contribution of the $k$th term is not affected by the
  update since its value does not depend on $s_i$ or $s_h$.

  \item For $k \in (i, h)$, the $k$th term can only increase. Indeed, for
  small enough $\varepsilon$ and for all $s'_i \in [s_i, \, s_i +
  \varepsilon]$ we have $\sum_{j=1}^k s'_j < 1 - \min_{1 \leq j \leq h } s'_j$
  and thus the contribution of the $k$th term to the value of $s'$ is
  \[\frac{s'_k}{1-\min_{1\leq j\leq k} s'_j} \geq \frac{s_k}{1-\min_{1\leq j
      \leq k} s_j},\] which in turn is its contribution to the value of $s$.

  \item
  For $k \in (h, r]$, the $k$th term does not change with the update because,
  since $\sum_{j=1}^k s'_j = \sum_{j=1}^k s_j$ and $\sum_{j=1}^h s_j \geq 1-
  \min_{1\leq j \leq h} s_j$, its contribution is always
  \[\frac{s_k} {\max\{\sum_{j=1}^k s_k,\, 1-\min_{h< j \leq k} s_j\}}.\]

  \item Regarding the $i$th term, for any $s'_i \in [s_i, s^*_i]$ we have
  $\sum_{j=1}^i s_j < 1 - \min_{1\leq j \leq i} s'_j = 1- s'_i$. Thus its
  contribution to the value of $s$ is $\frac{s'_i}{1-s'_i}$. Imagine
  increasing $s'_i$ continuously from $s_i$ to $s_i + \varepsilon$. The rate of
  change of its contribution to the value as a function of $s'_i$ is
  \[\frac{\partial}{\partial x}  \left(\frac{x}{1-x}\right)_{x =
    s'_i}\hspace{-1ex} =\ \, \frac{1}{(1-s'_i)^2}.\]

  \item Since the $h$th term decreases with the update, we need to show that
  its rate of change, as we decrease $s'_h$ from $s_h$ to $s_h - \varepsilon$, does
  not cancel out the rate of change of the $i$th term. Suppose $\sum_{j=1}^h
  s_j > 1-s_h$ then  its rate of change is
  \[-\, \frac{\partial}{\partial x} \left( \frac{x}{\sum_{j=1}^h s_j}
  \right)_{x = s'_h} \hspace{-1ex} =\ - \, \frac{1}{\sum_{j=1}^h s_j} \geq - \,
  \frac{1}{1-s_i} > - \, \frac{1}{(1-s_i)^2}.\] Let us consider what happens when $\sum_{j=1}^h s_j =
  1-s_h$. In this case $s_h \leq s_i$ since
  \(1-s_h = \sum_{j = 1}^h s_j \geq 1- \min_{1\leq j \leq h} s_j \geq 1-
  s_i.\)
  Thus, the rate of change of
  the $h$th term is
  \[-\, \frac{\partial}{\partial x} \left( \frac{x}{1-x} \right)_{x = s'_h}
  \hspace{-1ex} =\ - \, \frac{1}{(1-s'_h)^2}\ \geq\ - \frac{1}{(1-s_i)^2} .\]
  \end{itemize}

  We claim that for small enough $\varepsilon$, the value of $s'$ must be strictly greater
  than the value of $s$. Indeed, by the discussion above, the overall change in value
  is at least
  \[\int_{s_i}^{s_i + \varepsilon} \frac{1}{(1- s'_i)^2} - \frac{1}{(1-s_i)^2}\ \mathrm{d} s'_i > 0.\]

  Now let us see what happens when $s_i > s^*_i$. In this case we build our
  new solution $s'$ by decreasing $s_i$ and increasing $s_r$ by the same
  infinitesimally small amount $\varepsilon$. As before, terms before the $i$th
  do not depend on $s_i$ or $s_r$ and therefore are not affected by the
  update. Since $s_i > \sum_{j=i+1}^r s_j > s_{i+1}$ we have $\min_{1\leq j
    \leq k} s'_j = \min_{1\leq j \leq k } s_j$ for $k \in (i, r)$. Therefore,
  in this case, the $k$th term can only increase
  \[ \frac{s'_k}{\max\{\sum_{j = 1}^k s'_j, \, 1 - \min_{1\leq j \leq k }
    s'_j\}} \geq \frac{s_k}{\max\{\sum_{j = 1}^k s_j, \, 1 - \min_{1\leq j \leq k }
    s_j\}}. \]
  The $i$th term decreases and its rate of change is
  \[- \, \frac{\partial}{\partial x} \left(\frac{x}{1-2^{-i+1} + x}\right)_{x
    = s_i}\hspace{-1ex} = \ -\, \frac{1-2^{-i+1}}{(1-2^{-i+1} + s_i)^2} > -\,
  \frac{1-2^{-i+1}}{(1-2^{-i})^2} > - \, 1.\]
  On the other hand, the $r$th term increases and its rate of change is~1 due
  to our assumption that $\sum_i s_i = 1$. Therefore, the overall rate
  of change of value is strictly positive.

\subsection{Proof of Lemma \ref{lem:optimal-solution-param}}

  The plan is to show that given some solution $s$, either there exists $x \in
  \left[\frac{1}{t+1}, \frac{1}{t}\right]$ such that the solution $s^*$
  induced by $x$ equals $s$, or we can construct another solution $s'$ that is
  closer to $s^*$ and has value at least as large as $s$. This process is
  repeated until we converge to $s^*$.

  First, if $\sum_{j=1}^r s_j < 1$ then we can safely increase $s_r$ to until
  the bin is full. Note that we can always do this because there is no upper
  bound on $s_r$. From now on we assume that $\sum_{j=1}^r s_j = 1$.

  Suppose there exists $s_i < \frac{1}{t+1}$ for some $i < t$ and let $i$ be
  the smallest such index. Let $h$ be the smallest index such that
  $\sum_{j=1}^j s_j \geq \max\{1- \min_{1\leq j \leq h} s_j,
  \frac{t}{t+1}\}$. As was done in the proof of
  Lemma~\ref{lem:optimal-solution}, we increase $s_i$ and decrease $s_h$ by
  the same $\varepsilon$ amount. The same argument used before shows that the
  value of $s'$ is greater than the value of $s$. Therefore, we can assume
  that $s_i \geq \frac{1}{t+1}$ for all $i < t$. Under this assumption, each
  item contributes $s_i \frac{t}{t+1}$ to the objective, since \[1-\min_{1\leq
    j\leq i} s_i \leq 1- \frac{1}{t+1} = \frac{t}{t+1} \text{ and }\sum_{1
    \leq j \leq i} s_j \leq \frac{t-1}{t} < \frac{t}{t+1}.\] Setting $s'_i =
  \frac{\sum_{1\leq j < t} s_j}{t-1}$ for each $i \in [t-1]$ does not affect
  the contribution of these items and can only increase the contribution of
  the remaining items since the transformation does not change the total size,
  but may increase the minimum size of the first $t-1$ items. Therefore, we
  can assume that $s_1 = \cdots = s_{t-1} = x$ for some $x \in [\frac{1}{t+1},
  \frac{1}{t}]$.

  At this point, we can apply the exact same argument as the one used in the
  proof of Lemma~\ref{lem:optimal-solution}. For $i = t$ we note that if $s_t
  < s^*_t$ then for any $s'_i \in [s_i, s_i + \varepsilon]$ we have $1- s'_i >
  1- \frac{1-x\,(t-1)}{2} \geq \frac{t}{t+1},$ where the last inequality uses
  $x\geq \frac{1}{t+1}$, and $1-s'_i = 1- 2 s'_i + s'_i > x\, (t-1) + s'_i =
  \sum_{j = 1}^t s'_i.$ Therefore, the contribution of the $t$th term is
  $\frac{s'_t}{1-s'_t}$.  Similarly, if $s_t > s^*_t$ then the contribution is
  $\frac{s'_t}{\sum_{j=1}^t s'_j}$. These are the properties needed to apply
  the argument used before. The conclusion is that for all $t\leq i< r$ the
  value of the program is maximized by setting $s_i$ to $\frac{1 -
    \sum_{j=1}^{i-1} s_j}{2}$.

\subsection{Proof of Lemma \ref{lem:maximum-x}}

Consider the variable change $y = \frac{1}{1-(t-1) x}$, which for $t \geq 2$
  maps the range $[\frac1{t+1},\frac1t]$ for $x$ into the range $[\frac{t+1}2,
  t]$ for $y$:
\[\lambda^t(x) = g(y) = (1-1/y) \frac{t+1}{t} + \sum_{i=1}^{\infty} \frac{1}{y \, 2^i
  -1}.\]
This function and its derivative converge uniformly for $y$ in $[\frac{t+1}2,
t]$. Thus, the first derivative of $g$ can be obtained by term-wise
differentiation
\[g'(y) = \frac{t+1}{t \, y^{2}} - \sum_{i = 1}^{\infty} \frac{2^{i}}{(y\, 2^i
  - 1)^2} = \frac{1}{y^2} \left( \frac{t+1}{t} - \sum_{i =1}^\infty
  \frac{2^i}{(2^i - 1/y)^2} \right).\]
Notice that each term of the infinite sum, and thus the sum itself, is a
decreasing function of $y$ for $y\geq 1$. It follows that either the sign of
$g'$ is the same throughout the interval $[\frac{t+1}{2}, t]$ or it changes
from negative to positive. In either case, the maximum must be attained at one
of the ends of the interval. Hence, the maximum of $\lambda^t(x)$ in the
domain $[\frac{1}{t+1}, \frac1{t}]$ is attained either at $\frac{1}{t+1}$ or
at $\frac{1}{t}$.

We claim that $\lambda^t(\frac{1}{t+1}) > \lambda^t(\frac{1}{t})$ for all $t \geq 2$. Indeed, taking the difference of
these two values we get \begin{equation*}
  \lambda^t\left(\frac{1}{t+1}\right) - \lambda^t\left(\frac{1}{t}\right)
  = \frac{1}{t^2} - \sum_{i=1}^\infty \frac{1}{2^i\, t^2 +  (2^i-2)\, t
    - 1
    + 2^{-i}}
\end{equation*}

If the denominator of each term of the infinite sum were larger than $2^i t^2$
then it would immediately follow that the right hand side is always
positive. Unfortunately, this is not true for the first term. Nevertheless, it
is true for the remaining terms, and the first and second terms together are
less than $\frac{3}{4t^2}$. Therefore,
\begin{equation*}
  \lambda^t\left(\frac{1}{t+1}\right) - \lambda^t\left(\frac{1}{t}\right) > \frac{1}{t^2} -
  \sum_{i=1}^\infty \frac{1}{2^i\, t^2 } = 0
\end{equation*}

\subsection{Proof of Theorem \ref{thm:20}} \label{app:proof-thm-20}

Let $s>2$ be an integer. We define a construction with $s+1$ phases of indices $0 \leq j \leq s$, where the items of phase $j$ have sizes which are close to $\frac{1}{(t+1)\cdot 2^j}$, but can be slightly smaller or slightly larger than this value.
We let $OPT=t(t-1)\cdot n +n$, and assume that $n$ is a large enough integer, such that $n>2^{s^3}$, $n>>t$. We use a sequence of small values, $\delta_j$ such that $\delta_j=\frac{1}{(4n)^{3s-2j}}$. Note that this implies $\delta_{j+1}=(4n)^2\delta_j$ for $0\leq j\leq s-1$. For each $t\geq 2$, $t\in\mathbb{N}$ we use two sequences of positive integers $r^t_j\leq n$ and $d^t_j\leq n$, for $2 \leq j \leq s$, and in addition, $r^t_0=n$,\;$d^t_0=0$ and $r^t_1=\frac{n}{t+1}$,\;$d^t_1=\frac{t}{t+1}n$ (and thus $r^t_1+d^t_1=n$). We define $r^t_{j+1}=\frac{r^t_j-1}{(t+1)\cdot 2^{j-1}}$ and $d^t_{j+1}=r^t_{j}-r^t_{j+1}=\frac{((t+1)\cdot 2^{j-1}-1)r^t_j+1}{(t+1)\cdot 2^{j-1}}=((t+1)\cdot2^{j-1}-1)r^t_{j+1}+1$, for $1\leq j \leq s-1$.

\begin{observation}    \label{thm:21}
For each $1\leq j\leq s$, $\frac{n}{(t+1)^j\cdot 2^{j(j-1)/2}}-1\leq  r^t_j \leq \frac{n}{(t+1)^j\cdot 2^{j(j-1)/2}}$.
\end{observation}
\begin{proof}
For $j=1$ it holds by definition. We next prove the property for $r^t_{j+1}$
using the definition of the sequence $r^t_j$.  We have
$r^t_{j+1}=\frac{r^t_j-1}{(t+1)\cdot 2^{j-1}}$ for $j\geq1$. From this
definition, we get (by induction) that
\begin{align*}
  r^t_{j+1} & =r^t_1 \prod_{i=1}^j{\frac{1}{(t+1)\cdot
      2^{i-1}}}-\sum_{i=1}^{j-1}{\frac{1}{(t+1)\cdot
      2^{i}}}  \\
      &
      =\frac{r^t_1}{(t+1)^j}\cdot\frac{1}{2^{j(j-1)/2}}-\frac{1}{t+1}(1-\frac{1}{2^{j-1}})
      \\
      & < \frac{n}{(t+1)^j\cdot 2^{j(j-1)/2}},
\end{align*}
as $r^t_1< n$, and for $t\geq2$ $\frac{1}{t+1}(1-\frac{1}{2^{j-1}})>0$. On the
other hand,  $\frac{1}{(t+1)^{j-1}\cdot 2^{j(j-1)/2}}\leq 1$ holds, since $(t+1)^{j-1}\cdot2^{j(j-1)/2}\geq 1$
for $j\geq1$. So $r^t_{j+1}\geq \frac{n}{(t+1)^j\cdot 2^{j(j-1)/2}}-1$.
\end{proof}
The input set of items for $t\geq 2$ consists of multiple phases.
Phase 0 consists of the following sets of items; $nt$ items of size $\sigma_{01}=\frac{1}{t+1}+\Delta nt^2(t-1)+\Delta$, $t(t-1)n$ items of size $\sigma_{02}=\frac{1}{t+1}-\Delta nt(t-1)$, and pairs of items of sizes $\sigma^i_{03}=\frac{1}{t+1}+\Delta nt(t-1)+i\Delta$ and $\sigma^i_{04}=\frac{1}{t+1}-i\Delta$ for $1\leq i\leq t(t-1)n$, such that $\Delta=\frac{2\delta_0}{nt(t-1)+1}$. Note that $\sigma^i_{03}+\sigma^i_{04}=\frac{2}{t+1}+\Delta nt(t-1)$. There are also $(t-2)\cdot t(t-1)n$ items of size $\sigma_{05}=\frac{1}{t+1}$. For $1\leq j\leq s$, phase $j$ consists of the following $2d^t_j+r^t_j$ items. There are $r^t_{j}$ items of size $\sigma_j=\frac{1}{(t+1)\cdot 2^j}+2(d^t_j+1)\delta_j$, and for $1\leq i \leq d^t_j$, there are two items of sizes $\pi^i_{j}=\frac{1}{(t+1)\cdot 2^j}+(2i-1)\delta_j$ and $\theta^i_j=\frac{1}{(t+1)\cdot 2^j}-2i\delta_j$. Note that $\pi^i_{j}+\theta^i_j=\frac{1}{(t+1)\cdot 2^{j-1}}-\delta_j$.
A bin of level $j$ in the optimal packing contains only items of phases $1,\ldots, j$. A bin of level $s+1$ contains items of all phases.
The optimal packing contains $t(t-1)n$ bins of level 0, $d^t_j$ bins of level $j$, for $1 \leq j \leq s$, and the remaining bins are of level $s+1$.
Note that $\sum\limits_{j=1}^s d^t_j =\frac{t}{t+1}n+\sum\limits_{j=2}^s d^t_j=\frac{t}{t+1}n+r^t_1-r^t_s=\frac{t}{t+1}n+\frac{1}{t+1}n-r^t_s=n-r^t_s$. Thus, the number of level $s+1$ bins is (at most) $r^t_s$, and we have $n$ bins of levels $1 \leq j \leq s+1$ allocated, in addition to the $t(t-1)n$ bins of level 0. In total, the packing contains of at most $t(t-1)n+n=(t(t-1)+1)n$ bins. The optimal packing of the set of items specified above is defined as follows. A level 0 bin contains $t-2$ items of size $\sigma_{05}$, one item of size $\sigma_{02}$ and, in addition, one pair of items of sizes $\sigma^i_{03}$ and $\sigma^i_{04}$ for a given value of $i$ such that $1\leq i\leq t(t-1)n$.
For $1 \leq j \leq s$, a level $j$ bin contains $t$ items of size $\sigma_{01}$ and one item of each size $\sigma_k$ for $1\leq k \leq j-1$, and, also, one pair of items of sizes $\pi^i_{j}$ and $\theta^i_j$ for a given value of $i$ such that $1 \leq i \leq d^t_j$. A bin of level $s+1$ contains $t$ items of size $\sigma_{01}$ and one item of each size $\sigma_k$ for $1\leq k \leq s$.
\begin{claim}  \label{thm:23}
This set of items can be packed into $n+t(t-1)n$ bins, i.e., $OPT \leq (1+t(t-1))n$
\end{claim}
\begin{proof}
First, we show that every item was assigned into some bin.
Consider the $nt$ items of size $\sigma_{01}$. Each $t$-tuple of these items is assigned into a bin of level $1\leq j\leq s$ together.
Consider items of size $\pi^i_{j}$ and $\theta^i_j$. Such items exist for $1\leq i \leq d^t_j$, therefore, every such pair is assigned into a bin (of level $1\leq j\leq s$) together.
Next, consider items of size $\sigma_j$ for some $1\leq j \leq s$. The number of such items is $r^t_j$. The number of bins which received such items is $\sum\limits_{k=j+1}^{s} d^t_k+r^t_s=r^t_j$.
As to the items of size $\sigma_{02}$. There are $t(t-1)n$ such items, each item is assigned into one of the $t(t-1)n$
bins of level 0.
The items $\sigma^i_{03}$ and $\sigma^i_{04}$ that exist for $1\leq i\leq t(t-1)n$. Every such pair is assigned into one of the $t(t-1)n$ level 0 bins together.
And, finally consider the $(t-2)\cdot t(t-1)n$ items of size $\sigma_{05}$. Each $(t-2)$ tuple of these items is assigned into one of the $t(t-1)n$ level 0 bins.

We further show that the sum of sizes of items in each bin does not exceed 1. Consider a bin of level 0. The sum of items it contains is: $(t-2)\sigma_{05}+\sigma_{02}+\sigma^i_{03}+\sigma^i_{04}=(t-2)\cdot \frac{1}{t+1}+\frac{1}{t+1}-\Delta nt(t-1)+\frac{2}{t+1}+\Delta nt(t-1)=1$.
Now, consider a bin of level $j$ for some $1 \leq j \leq s$. The sum of items
packed in it is:
\begin{align*}
  t\cdot \sigma_{01}+ & \sum\limits_{k=1}^{j-1}\sigma_{k}+\frac{1}{(t+1)\cdot
  2^{j-1}}-\delta_j \\ &=t\cdot (\frac{1}{t+1}+\Delta
nt^2(t-1)+\Delta)+\sum\limits_{k=1}^{j-1} (\frac{1}{(t+1)\cdot
  2^k}+2(d^t_k+1)\delta_k)+\frac{1}{(t+1)\cdot
  2^{j-1}}-\delta_j  \\ & =\frac{t}{t+1}+t\cdot (\Delta
nt^2(t-1)+\Delta)+\frac{1}{(t+1)\cdot
  2^{j-1}}-\delta_j+\frac{1}{t+1}\sum\limits_{k=1}^{j-1}
(\frac{1}{2^k}+2(d^t_k+1)\delta_k) \\ & =\frac{t}{t+1}+\frac{1}{(t+1)\cdot
  2^{j-1}}+\frac{1-(\frac{1}{2})^{j-1}}{t+1}+t^2 \cdot
2\delta_0-\delta_j+\sum\limits_{k=1}^{j-1} 2(d^t_k+1)\delta_k\\ &=1+t^2 \cdot
2\delta_0+\sum\limits_{k=1}^{j-1} 2(d^t_k+1)\delta_k-\delta_j.
\end{align*}
It is left to show that $t^2\cdot 2\delta_0+\sum\limits_{k=1}^{j-1} 2(d^t_k+1)\delta_k-\delta_j\leq 0$ holds.
As $d_k+1 \leq n$ and $\delta_j$ is a strictly increasing sequence, we have $2(d_k+1)\delta_k \leq 2n\delta_{j-1}$, and since $j-1\leq s<n$, $\sum\limits_{k=1}^{j-1} 2(d_k+1)\delta_k < 4n^2\delta_{j-1}$. Also, as $t<n$, $t^2 \cdot 2\delta_0< 2n^2\delta_{j-1}$. Using $\delta_j=16n^2 \delta_{j-1}$ we get that the sum $t^2\cdot 2\delta_0+\sum\limits_{k=1}^{j-1} 2(d^t_k+1)\delta_k$ is smaller than $\delta_j$.

It is left to consider a bin of level $s+1$. The sum of items in it is:
\begin{align*}
t\cdot \sigma_{01}+ \sum\limits_{k=1}^s \sigma_k & = t\cdot (\frac{1}{t+1}+\Delta
nt^2(t-1)+\Delta)+\sum\limits_{k=1}^s (\frac{1}{(t+1)\cdot
  2^k}+2(d^t_k+1)\delta_k) \\  & =\frac{t}{t+1}+t\cdot (\Delta
nt^2(t-1)+\Delta)+\frac{1-(\frac{1}{2})^{s}}{(t+1)}+\sum\limits_{k=1}^{s}
2(d_k+1)\delta_k \\  & =1-\frac{(\frac{1}{2})^{s}}{(t+1)}+t\cdot
2\delta_0+\sum\limits_{k=1}^{s} 2(d^t_k+1)\delta_k.
\end{align*}
We have $2(d^t_k+1)\delta_k \leq 2n\delta_s= \frac{1}{2^{2s-1}n^{s-1}}$. Since
$1<s<n$, $t<n$ and $t\cdot 2\delta_0< 2n^2\delta_{s}$, we get that the
quantity above is at most
\begin{align*}
  1-\frac{(\frac{1}{2})^{s}}{(t+1)}+\frac{n}{2^{2s-1}n^{s-1}}+2n^2\delta_{s} &
  =1-\frac{1}{2^{s}(t+1)}+
\frac{1}{2^{2s-1}n^{s-2}}+2n^2\delta_{s} \\ & =1-\frac{1}{2^{s}(t+1)}+
\frac{1}{2^{2s-1}n^{s-2}}+\frac{2n^2}{(4n)^s} \\  &=1-\frac{1}{2^{s}(t+1)}+
\frac{1}{2^{2s-1}n^{s-2}}+\frac{1}{2^{2s-1}n^{s-2}} \\ & =1-\frac{1}{2^{s}(t+1)}+
\frac{1}{2^{2(s-1)}n^{s-2}}<1.
\end{align*}
\end{proof}

Before introducing the \textit{NE} packing for this set of items, we slightly modify the input by removing a small number of items. Clearly, $OPT \leq (1+t(t-1))n$ would still hold for the modified input.
The modification applied to the input is a removal of items $\pi^1_{j}$ and $\theta^{d^t_j}_j$ for all $1 \leq j \leq s$, the two items $\sigma^1_{03}$ and $\sigma^{t(t-1)n}_{04}$ and $(t-2)$ of the $\sigma_{05}$ items from the input.
We now define an alternative packing, which is a \textit{NE}. There are three types of bins in this packing.
The bins of the first type are bins with items of phase $j$, $1\leq j\leq s+1$.
We construct $r^t_j$ such bins. A bin of phase $j$ consists of $(t+1)\cdot 2^{j}-1$ items, as follows. One item of size $\sigma_j=\frac{1}{(t+1)\cdot 2^j}+2(d^t_j+1)\delta_j$, and $(t+1)\cdot 2^{j-1}-1$ pairs of items of phase $j$. A pair of items of phase $j$ is defined to be the items of sizes $\pi^{i+1}_{j}$ and $\theta^i_j$, for some $1 \leq i \leq d^t_j-1$. The sum of sizes of this pair of items is $\frac{1}{(t+1)\cdot 2^j}+(2i+1)\delta_j+\frac{1}{(t+1)\cdot 2^j}-2i\delta_j=\frac{2}{(t+1)\cdot 2^{j}}+\delta_j=\frac{1}{(t+1)\cdot 2^{j-1}}+\delta_j$.

Using $d^t_j=((t+1)\cdot 2^{j-1}-1)r^t_j+1$ we get that all phase $j$ items, for $1\leq j\leq s$ are packed. The sum of items in every such bin is $1-\frac{1}{(t+1)\cdot 2^{j-1}}+((t+1)\cdot 2^{j-1}-1)\delta_j+ \frac{1}{(t+1)\cdot 2^j}+2(d^t_j+1)\delta_j=1-\frac{1}{(t+1)\cdot 2^{j}}+\delta_j((t+1)\cdot 2^{j-1}+1+2d^t_j)$.

The $nt$ bins of the second type in the \textit{NE} packing contain $(t-1)$
items of size $\sigma_{02}=\frac{1}{t+1}-\Delta nt(t-1)$ and one item of size
$\sigma_{01}=\frac{1}{t+1}+\Delta nt^2(t-1)+\Delta$, from the 0 phase
bins. The load of each such bin is
\begin{align*}
  (t-1) \left(\frac{1}{t+1}-\Delta  nt(t-1)\right)+& \frac{1}{t+1}+\Delta
  nt^2(t-1)+\Delta \\ & =\frac{t}{t+1}-\Delta nt(t-1)^2+\Delta
  nt^2(t-1)+\Delta \\ & =\frac{t}{t+1}+\Delta
  nt(t-1)(t-(t-1))+\Delta \\ & =\frac{t}{t+1}+\Delta
  nt(t-1)+\Delta \\ & =\frac{t}{t+1}+\Delta(nt(t-1)+1)\\ & =\frac{t}{t+1}+2\delta_0,
\end{align*}
by definition of $\Delta$.
As there are in total $t(t-1)n$ identical items of size $\sigma_{02}$ and $nt$ identical $\sigma_{01}$ items in the input set, we get that all these items are packed in these $nt$ second type bins in the \textit{NE} packing constructed above.

The $t(t-1)n-1$ bins of third type in the \textit{NE} packing each contain $(t-2)$ items of size $\sigma_{05}=\frac{1}{t+1}$, and, in addition, one pair of items of sizes $\sigma^{i+1}_{03}$ and $\sigma^i_{04}$, for some $1\leq i\leq t(t-1)n$ from the phase 0 bins. The sum of sizes of this pair of items is:
$\sigma^{i+1}_{03}+\sigma^i_{04}=\frac{1}{t+1}+\Delta nt(t-1)+(i+1)\Delta+\frac{1}{t+1}-i\Delta=\frac{2}{t+1}+\Delta(nt(t-1)+1)=\frac{2}{t+1}+2\delta_0$.
Thus, the total load of such bin is $(t-2)\cdot \frac{1}{t+1}+\frac{2}{t+1}+2\delta_0=\frac{t}{t+1}+2\delta_0$, which equals the load of the bins of the second type in the \textit{NE} packing.
As there are in total $((t-2)\cdot t(t-1)n-(t-2))=(t-2)(t(t-1)n-1)$ items of size $\sigma_{05}$ and $t(t-1)n-1$ pairs of $\sigma^i_{03}$ and $\sigma^i_{04}$ items, we conclude that all the items of size $\sigma_{05}$ and $\sigma^i_{03}$, $\sigma^i_{04}$ are packed in these $t(t-1)n-1$ \textit{NE} bins of the third type, as defined above.

We now should verify that the sum of sizes of the items packed in the three types of bins in the defined \textit{NE} packing does not exceed 1.
This holds for the second and the third type bins, as:  $\frac{t}{t+1}+2\delta_0<\frac{t}{t+1}+2n\delta_s=\frac{t}{t+1}+\frac{2n}{(4n)^s}=\frac{t}{t+1}+
\frac{1}{2^{2s-1}n^{s-1}}< \frac{t}{t+1}+\frac{1}{t+1}=1$.
For the bins of the first type, this property directly follows from the inequality proven in the next claim.

\begin{claim}  \label{thm:24}
The loads of the bins in the packing defined above are monotonically increasing as a function of the phase.
\end{claim}
\begin{proof}
It is enough to show $1-\frac{1}{(t+1)\cdot 2^{j}}+\delta_j((t+1)\cdot 2^{j-1}+1+2d^t_j)<1-\frac{1}{(t+1)\cdot 2^{j+1}}$ for $1\leq j\leq s$, $t\geq 2$ which is equivalent to proving $\delta_j((t+1)\cdot 2^{j-1}+1+2d^t_j)2^{j+1}<\frac{1}{t+1}$.
Using $d^t_j<n$, we have: $\delta_j((t+1)\cdot 2^{j-1}+1+2d^t_j)2^{j+1}<\delta_j((t+1)\cdot 2^{2j}+2^{j+2}n)<(t+1)\cdot 2\delta_j n^2$, as $n>2^{s^3}$. Using $\delta_j \leq \delta_s
=\frac{1}{2^{2s}n^s}\leq \frac{1}{16n^3(t+1)^2}$ we get $2\delta_j n^2<\frac{1}{t+1}$.

For $j=0$, $\frac{t}{t+1}+2\delta_0<1-\frac{1}{(t+1)\cdot 2^{j}}+\delta_j((t+1)\cdot 2^{j-1}+1+2d^t_j)$ holds for all $j\geq 1$, as $2\delta_0\leq \delta_j((t+1)\cdot 2^{j-1}+1+2d^t_j)$, since $t\geq 2$ and $\delta_j$ is a strictly increasing sequence.
\end{proof}

\begin{claim}  \label{thm:25}
The packing defined above is a valid NE packing.
\end{claim}
\begin{proof}
To show that this is a \textit{NE} packing, we need to show the an item of phase $j>0$ cannot migrate to a bin of a level $k \geq j$, since this would result in a load larger than 1, and that it cannot migrate to a bin of phase $k<j$, since this would result in a load smaller than the load of a phase $j$ bin. Due to the monotonicity we proved in Claim \ref{thm:24}, we only need to consider a possible migration of a phase $j$ item into a phase $j$ bin, and a phase $j-1$ bin, if such bins exist. Moreover, in the first case it is enough to consider the minimum size item and in the second case, the maximum size item of phase $j$.

For phase 0 items, since the smallest phase 0 item has size $\frac{1}{t+1}-\Delta nt(t-1)$, if it migrates to another bin of this phase, we get a total load of $\frac{t}{t+1}+\Delta(nt(t-1)+1)+\frac{1}{t+1}-\Delta nt(t-1)=1+\Delta>1$, as $\Delta>0$.

For items of phase $j\geq 1$:
The smallest phase $j$ item has size $\frac{1}{(t+1)\cdot2^j}-\delta_j(2(d^t_j-1))= \frac{1}{(t+1)\cdot2^j}-\delta_j(2d^t_j-2)$. If it migrates to another bin
of this phase, we get a total load of
\begin{align*} 1-\frac{1}{(t+1)\cdot 2^{j}} +\delta_j( & (t+1)\cdot
    2^{j-1}+1+2d^t_j)+ \frac{1}{(t+1)\cdot2^j}-\delta_j(2d^t_j-2) \\ &
  =1+\delta_j((t+1)\cdot 2^{j-1}+1+2d^t_j)-2d^t_j\delta_j+2\delta_j \\ &
  =1+\delta_j(3+(t+1)\cdot2^{j-1})>1.
\end{align*}

The check for the largest item in the phase should be done separately for cases $j=1$ and $j\geq2$, because we want to show that the largest item of phase $j=1$ (in first type bin) cannot migrate into a phase 0 bin (a second or third type bin), while for the largest item of phase $j\geq2$ we need to show that it cannot move into other bin of first type.
For phase $j=1$: The largest phase item has size $\frac{1}{2(t+1)}+2(d^t_1+1)\delta_1$. If it migrates to a bin of phase 0, we get a load of $\frac{t}{t+1}+2\delta_0+\frac{1}{2(t+1)}+2(d^t_1+1)\delta_1=\frac{2t+1}{2(t+1)}+2\delta_0+
2(d^t_1+1)\delta_1$. This load is strictly smaller than a load of level $1$ which is $1-\frac{1}{(t+1)\cdot2}+\delta_1((t+1)+1+2d^t_1)=\frac{2t+1}{2(t+1)}+\delta_1((t+1)+1+2d^t_1)$, as $t\geq2$ and $\delta_1>\delta_0$.

For phase $j\geq2$: The largest phase $j$ item has size
$\frac{1}{(t+1)\cdot2^j}+2(d^t_j+1)\delta_j$. If it migrates to a bin of phase
$j-1$, we get a load of
\begin{align*}
  1-\frac{1}{(t+1)\cdot 2^{j-1}}+\delta_{j-1} & ((t+1)\cdot
  2^{j-2}+1+2d^t_{j-1})+\frac{1}{(t+1)\cdot
    2^j}+2(d^t_j+1)\delta_j \\ & =1-\frac{1}{(t+1)\cdot 2^j}+\delta_{j-1}((t+1)\cdot
  2^{j-2}+1+2d^t_{j-1})+2(d^t_j+1)\delta_j\\ &=1-\frac{1}{(t+1)\cdot
    2^j}+\delta_{j-1}((t+1)\cdot
  2^{j-2}+1+2d^t_{j-1})+2d^t_j\delta_j+2\delta_j.
\end{align*}
We compare this load with $1-\frac{1}{(t+1)\cdot2^j}+\delta_j((t+1)\cdot 2^{j-1}+1+2d^t_j)$, and prove that the first load is smaller. Indeed $\delta_{j-1}((t+1)\cdot 2^{j-2}+1+2d^t_{j-1}) < \delta_j((t+1)\cdot2^{j-1}-1)$ since $\delta_{j}=16n^2 \delta_{j-1}$, $n>2^{s^3}$ and $((t+1)\cdot 2^{j-2}+1+2d^t_{j-1}) <4n(t+1) < 16n^2 ((t+1)\cdot 2^{j-1}-1)$.
\end{proof}

Finally, we bound the \textit{PoA} as follows. The cost of the resulting \textit{NE} packing is $nt+t(t-1)n-1+\sum\limits_{j=1}^s r^t_j=t^2n-1+\sum\limits_{j=1}^s r^t_j$. Using Observation \ref{thm:21} we get that $\sum\limits_{j=1}^s r^t_j \geq \sum\limits_{j=1}^s (\frac{n}{(t+1)^j\cdot 2^{j(j-1)/2}}-1 ) $ and since $OPT=t(t-1)\cdot n +n$ and $n>>s$, we get a ratio of at least \[\frac{t^2+\sum_{j=1}^s {(t+1)^{-j}\cdot 2^{-j(j-1)/2}}}{t(t-1)+1}.\] Letting $s$ tend to infinity as well results in the claimed lower bound.

Note that we assume that all numbers $r^t_j$ and $d^t_j$ are integer values for each $t\geq 2$, which is not necessarily the case. To overcome this, we let $r^t_{j+1}=\big \lfloor \frac{r^t_j-1}{(t+1)\cdot 2^{j-1}} \big \rfloor$, for $0\leq j \leq s-1$, and $d^t_{j+1}=((t+1)\cdot2^{j-1}-1)r^t_{j+1}+1$. In this case, it is possible to prove $\frac{n}{(t+1)^j\cdot 2^{j(j-1)/2}}-3 \leq  r^t_j \leq \frac{n}{(t+1)^j\cdot 2^{j(j-1)/2}}$, which leads to the same result.

\subsection{Proof of Claim \ref{thm:40}}

Consider the well-known First Fit algorithm (FF for short) for bin packing. FF packs each item in turn into the lowest indexed bin to where it fits. It opens a new bin only in the case where the item does not fit into any existing bin. It was shown in \cite{JohnsonDUGG74} that any bin (accept for maybe two) in the packing produced by FF is more than $\frac{t}{t+1}$ full for any $t\geq 2$. For each $N_G$ instance, it is possible to define (modulo reordering the items) an instance for which running the FF algorithm will produce exactly the packing $b$. So, as any \textit{NE} packing $b$ can be produced by a run of FF, it has all the properties of a FF packing, including the one mentioned above.

\subsection{Proof of Claim \ref{thm:41}}

First, consider the bins in group $\mathcal{D}$. For $t\geq 3$, as all bins in $\mathcal{D}$ are filled by no more than $\frac{t^2-t+1}{t^2}$, no bin in this group (except maybe the leftmost bin) contains an item of size in $(0,\frac{t-1}{t^2}]$, as such an item will reduce its cost by moving to the leftmost bin in $\mathcal{D}$ (which is the bin with the largest load in $\mathcal{D}$), contradicting the fact that $b$ is an $\textit{NE}$.
Hence, all the items in bins (except for maybe one) in group $\mathcal{D}$  have items of sizes in $(\frac
{t-1}{t^2}, \frac{1}{t}]$. For $t=2$, as all bins in $\mathcal{D}$ are filled by no more than $\frac{17}{24}$, no bin in this group (except maybe the leftmost bin) contains an item of size in $(0,\frac{7}{24}]$, as such an item will reduce its cost by moving to the leftmost bin in $\mathcal{D}$, which contradicts the fact that $b$ is an $\textit{NE}$. Hence, all the items in bins (except for maybe one) in group $\mathcal{D}$  have items of sizes in $(\frac
{7}{24}, \frac{1}{2}]$.

Now, consider the bins in group $\mathcal{C}$. For $t\geq 3$, as all bins in $\mathcal{C}$ are filled by no more than
$\frac{t+1}{t+2}$, no bin in this group (except maybe the leftmost bin) contains an item of size in $(0,\frac{1}{t+2}]$, as such an item will reduce its cost by moving to the leftmost bin in $\mathcal{C}$ (which is the bin with the largest load in $\mathcal{C}$), contradicting the fact that $b$ is an $\textit{NE}$. Also, no bin in $\mathcal{C}$ contains an item of size $x\in(\frac{1}{t+2}, \frac{t-1}{t^2}]$, as such an item will benefit from moving to a bin in group $\mathcal{D}$, as $x+\frac{t}{t+1}>\frac{t+1}{t+2}$ for any $x>\frac{1}{(t+2)}$. Hence, all the items in bins in group $\mathcal{C}$  are of sizes in $(\frac{t-1}{t^2}, \frac{1}{t}]$. For $t=2$, as all bins in $\mathcal{C}$ are filled by no more than $\frac{3}{4}$, no bin in this group (except maybe the leftmost bin) contains an item of size in $(0,\frac{1}{4}]$, as such an item will reduce its cost by moving to the leftmost bin in $\mathcal{D}$, which contradicts the fact that $b$ is an $\textit{NE}$. Also, no bin in $\mathcal{C}$ contains an item of size $x\in(\frac{1}{4}, \frac{7}{24}]$, as such an item will benefit from moving to a bin in group $\mathcal{D}$, as $x+\frac{2}{3}>\frac{3}{4}$ for any $x>\frac{1}{4}$. Hence, all the items in bins (except for maybe one) in group $\mathcal{C}$  have sizes in $(\frac{7}{24}, \frac{1}{2}]$.

Finally, consider the bins in group $\mathcal{B}$. For $t\geq 3$, as all bins in $\mathcal{B}$ are filled by no more than $\frac{2t+1}{2(t+1)}$, no bin in this group (except maybe the leftmost bin) contains an item of size in $(0,\frac{1}{2(t+1)}]$, as such an item will reduce its cost by moving to the leftmost bin in $\mathcal{B}$ (which is the bin with the largest load in $\mathcal{B}$), contradicting the fact that $b$ is an $\textit{NE}$. Also, no bin in $\mathcal{B}$ contains an item of size $x\in(\frac{1}{2(t+1)}, \frac{t-1}{t^2}]$, as such an item will benefit from moving to a bin in group $\mathcal{D}$, as $x+\frac{t}{t+1}>\frac{2t+1}{2(t+1)}$ for any $x>\frac{1}{2(t+1)}$. Hence, all the items in bins (except for maybe one) in group $\mathcal{C}$  have items of sizes in $(\frac{t-1}{t^2}, \frac{1}{t}]$. For $t=2$, as all bins in $\mathcal{B}$ are filled by no more than $\frac{5}{6}$, no bin in this group (except maybe the leftmost bin) contains an item of size in $(0,\frac{1}{6}]$, as such an item will reduce its cost by moving to the leftmost bin in $\mathcal{B}$, which contradicts the fact that $b$ is an $\textit{NE}$. Also, no bin in $\mathcal{B}$ (except maybe the leftmost bin) contains an item of size $x\in(\frac{1}{6}, \frac{7}{24}]$, as such an item will benefit from moving to a bin in group $\mathcal{D}$, as $x+\frac{2}{3}>\frac{5}{6}$ for any $x>\frac{1}{6}$. Hence, all the items in bins (except for maybe one) in group $\mathcal{C}$ have sizes in $(\frac{7}{24}, \frac{1}{2}]$.

We conclude, that any bin in groups $\mathcal{B}$, $\mathcal{C}$ and $\mathcal{D}$, except for maybe a constant number of bins, contain only items of sizes in $(\frac{t-1}{t^2}, \frac{1}{t}]$ for $t\geq 3$, and items of sizes in $(\frac{7}{24}, \frac{1}{2}]$ for $t=2$.

Now, we show that each one of these bins contains exactly $t$ such items.
Note, that by definition of the groups all bins in $\mathcal{B}$, $\mathcal{C}$ and $\mathcal{D}$ (except maybe two)
have loads in $(\frac{t}{t+1},\frac{2t+1}{2(t+1)}]$ for $t\geq 3$, or in $(\frac 2 3, \frac 5 6]$ for $t=2$.

If a bin contains at most $t-1$ such items, then it has a load of at most $(t-1)\cdot \frac{1}{t}=\frac{t-1}{t}$ for $t\geq 3$ of at most $\frac{7}{24}$ for $t=2$, which is less than the assumed load in these bins, so they must have more than $(t-1)$ such items.

If a bin contains at least $t+1$ such items, then it has a load of at least $(t+1)\cdot \frac{t-1}{t^2}=1-\frac{1}{t^2}$, which is greater than $\frac{2t+1}{2(t+1)}$ for $t\geq 3$, or at least $\frac{7}{8}$
which is greater than $\frac 5 6$ for $t=2$, so they must have less than $(t+1)$ such items.

We conclude that each  bin in groups $\mathcal{B}$, $\mathcal{C}$ and $\mathcal{D}$), except for maybe 5 special bins (the leftmost bins in groups $\mathcal{B}$, $\mathcal{C}$ and $\mathcal{D}$ and the two rightmost bins in $\mathcal{D}$) contain exactly $t$ items with sizes in $(\frac{t-1}{t^2}, \frac{1}{t}]$ for $t\geq 3$, or exactly $2$ items of sizes in $(\frac{7}{24}, \frac{1}{2}]$ for $t=2$.

\subsection{Proof of Claim \ref{thm:50}}

Assume by contradiction that $(N_t+k)\cdot t$  of these items for $k\geq 0$ are packed together in $t$-tuples in bins of groups $\mathcal{B}$, $\mathcal{C}$ and $\mathcal{D}$ in the $NE$ packing. Consider the first $N_t$ such bins.
Call them $B_1, B_2,\ldots, B_{N_t}$. In a slight abuse of notation, we use $B_i$ to indicate both the $i$-th bin and its load. Denote the sizes of the remaining $N_t$ $t$-items by $t_1, t_2,\ldots,t_{N_t}$. These items are also packed in bins of groups $\mathcal{B}$, $\mathcal{C}$ and $\mathcal{D}$ in $b$, and share their bin with  $t-1$ $t$-items (when at least one of these items is not packed in any of the aforementioned $N_t$ bins in the optimal packing). Obviously, as all these $N_t\cdot (t+1)$ $t$-items fit into $N_t$ unit-capacity bins, $t_1+\ldots+ t_{N_t}+B_1+\ldots +B_{N_t} \leq N_t$ holds.
To derive a contradiction, we use the following observation:
\begin{observation}
A $t$-tuple of items with sizes in $(\frac{t-1}{t^2}, \frac 1 t]$ always has a greater total size than any $(t-1)$-tuple of such items.
\end{observation}
\begin{proof}
The total size of any $(t-1)$ items with sizes in $(\frac{t-1}{t^2}, \frac 1 t]$ is at most $\frac{t-1}{t}$, while the total size of any $t$ items with sizes in $(\frac{t-1}{t^2}, \frac 1 t]$ is strictly greater than  $\frac{t(t-1)}{t^2}=\frac{t-1}{t}$.
\end{proof}
Thus, any item $t_i$, $1\leq i\leq N_t$ would be better off sharing a bin with other $t$ items of size in $(\frac{t-1}{t^2}, \frac 1 t]$ instead of just $t-1$ such items as it does in the $NE$ packing $b$. For an item which shares a bin with $t-1$ $t$-items we conclude that the only reason it does not move to another bin with $t$ such items in $b$ is that it does not fit there.

So, we know that no item $t_1$, $1\leq i\leq N_t$ fits in any of the bins $B_1, B_2,\ldots, B_{N_t}$ in $b$.
We get that for any $1\leq i\leq N_t$, for any $1\leq j\leq N_t$, the inequality $t_i+B_j>1$ holds. Summing these inequalities over all $1\leq i\leq N_t$ and $1\leq j\leq N_t$ we get $t_1+\ldots+ t_{N_t}+B_1+\ldots +B_{N_t} > N_t$, which is a contradiction.

\end{document}